\documentclass[10pt, one column, draft cls]{IEEEtran}

\usepackage{tikz}
\usetikzlibrary{positioning}
\usetikzlibrary{arrows}
\usepackage{verbatim}
\usepackage{cuted}
\usepackage{flushend}

\newcommand{\numofcolumns}{1}
\newcommand{\isarxiv}{1}

\usepackage{amsthm}

\usepackage{amsmath,amssymb,bbm,color,psfrag,amsfonts,mathtools}
\usepackage{dsfont}
\usepackage{graphicx,cuted}                       
\usepackage[pdftex,pdfpagelabels,bookmarks,hyperindex,hyperfigures]{hyperref}
\usepackage{wrapfig}
\usepackage{pdfpages}
\usepackage[ruled,lined]{algorithm2e} 
\usepackage{algpseudocode}
   \SetKwInOut{Input}{input}\SetKwInOut{Output}{output}
\usepackage{accents}
\usepackage{lipsum}
\usepackage[T1]{fontenc}
\usepackage[utf8]{inputenc}
\usepackage{mathtools}
\usepackage{multicol}
\usepackage{paralist}   
\usepackage{comment}
\usepackage{kbordermatrix}
\usepackage{tikz}

\usepackage{booktabs}  
\usepackage{multirow}
\usepackage[normalem]{ulem}

\usepackage{caption}
\usepackage{subcaption}

\usepackage{cite}

\newcounter{problem}

\newtheorem{theorem}{Theorem}

\newtheorem{proposition}{Proposition}

\newtheorem{assumption}{Assumption}

\newtheorem{remark}{Remark}
\newtheorem{example}{Example}

\usepackage{color}

\newcommand{\real}{\mathbb{R}}

\newcommand{\RR}{\mathbb{R}}
\newcommand{\NN}{\mathbb{N}}
\newcommand{\PP}{\mathbb{P}}
\newcommand{\mc}{\mathcal}

\newcommand{\setdef}[2]{\left\{#1 \; | \; #2\right\}}
\newcommand{\zerobf}{\mathbf{0}}
\newcommand{\onebf}{\mathbf{1}}

\newcommand{\map}[3]{#1: #2 \rightarrow #3}

\newcommand{\congfunc}{\ell}
\newcommand{\diag}{\mathbf{diag}}

\newcommand{\prior}{\mu_0}
\newcommand{\state}{\omega}
\newcommand{\allstates}{\Omega}

\newcommand{\signal}{\pi}

\newcommand{\pfrac}{\nu}
\newcommand{\simplex}{\mc P}

\newcommand{\nlinks}{n}
\newcommand{\npaths}{\nlinks}
\newcommand{\nmesgs}{m}

\newcommand{\nstates}{s}
\newcommand{\E}{\mathbb{E}}
\newcommand{\posterior}{\mu}

\newcommand{\allsignals}{\Pi}

\newcommand{\supp}{\mathrm{supp}}

\newcommand{\de}{\mathrm{d}}
\newcommand{\pubsignal}{\signal^{\text{pub}}}
\newcommand{\indsignal}{\signal^{\text{ind}}}

\newcommand{\mom}{\eta}

\newcommand{\rank}{\mathrm{rank}}

\graphicspath{{../results/} , {../../../../fig/},{../fig/}}

\algnewcommand{\algorithmicgoto}{\textbf{go to}}%
\algnewcommand{\Goto}[1]{\algorithmicgoto~\ref{#1}}%

\usepackage{tikz}
\usepackage{pgfplots}
\usetikzlibrary{positioning}
\usetikzlibrary{shapes, arrows.meta, arrows}
\usepackage{relsize}
\usepackage{standalone}

\tikzset{fontscale/.style = {font=\relsize{#1}}
    }

\definecolor{nodecolor}{rgb}{255, 255, 255}
\definecolor{sourcenodecolor}{rgb}{255, 0, 0}
\definecolor{sinknodecolor}{rgb}{0, 255, 0}

\tikzset{ dot/.style ={circle, draw, inner sep=2pt},  
	     supply dot/.style ={circle, draw, inner sep=2pt, fill = sourcenodecolor}, 
	     demand dot/.style={circle, draw, inner sep=2pt, fill = sinknodecolor},
              main node/.style={circle,draw,font=\sffamily\bfseries\small, fill=nodecolor, line width = 0.5mm}, 
              supply node/.style={circle,draw,font=\sffamily\bfseries\small, fill=sourcenodecolor, line width = 0.5mm}, 
              demand node/.style={circle,draw,font=\sffamily\bfseries\small, fill=sinknodecolor, line width = 0.5mm},
              edge label/.style={font=\sffamily\small},
              main edge/.style={thick, auto},
              directed edge/.style={-Stealth,thick, auto}, 
              cascading edge/.style={line width=0.75mm, auto},
              infeasible node/.style={circle,minimum size=0.6cm, inner sep=0pt, draw, fill= red!50}, 
              feasible node/.style={circle,minimum size=0.6cm, inner sep=0pt, draw, fill= blue!50},
              face/.style={circle,draw,inner sep=1pt, fill=black}}

\newcommand{\revisionchange}[1]{#1}

\title{\revisionchange{Information Design in Non-atomic Routing Games with Partial Participation: \\ Computation and Properties}}

\author{Yixian Zhu\thanks{The authors are with the University of Southern California, Los Angeles, CA. \texttt{\{yixian,ksavla\}@usc.edu}. This work was supported in part by NSF CAREER ECCS \# 1454729 and CALTRANS \# MT-19-06 TO-017. 
K. Savla has financial interest in Xtelligent, Inc.
} \qquad Ketan Savla}

\date{}

\begin{document}
\maketitle

\begin{abstract}
We consider a routing game among non-atomic agents where link latency functions are conditional on an uncertain state of the network. The agents have the same prior belief about the state, but only a fixed fraction receive private route recommendations or a common message, which are generated by a known randomization,  referred to as \emph{private} or \emph{public signaling policy} respectively. The remaining agents choose route according to Bayes Nash flow with respect to the prior. 
We develop a computational approach to solve the optimal information design problem, i.e., to minimize expected social latency over all public or \emph{obedient} private signaling policies. 
For a fixed flow induced by non-participating agents, design of an optimal private signaling policy is shown to be a generalized problem of moments for polynomial link latency functions, and to admit an atomic solution with a provable upper bound on the number of atoms. This implies that, for polynomial link latency functions, information design can 
 be equivalently cast as a polynomial optimization problem. This in turn can be arbitrarily lower bounded by a known hierarchy of semidefinite relaxations. The first level of this hierarchy is shown to be exact for the basic two link case with affine latency functions. We also identify a class of private signaling policies over which the optimal social cost is non-increasing with increasing fraction of participating agents \revisionchange{for parallel networks}. This is in contrast to existing results where the cost of participating agents under a \emph{fixed} signaling policy may increase with their increasing fraction. 
%
%
\end{abstract}

\section{Introduction}
%
Route choice decision in traffic networks under uncertain and dynamic environments, such as the ones induced by recurring unpredictable incidents, can be a daunting task for agents. Private route recommendation or public information systems could therefore play an important role in such settings. While the agents have prior about the uncertain state, e.g., through experience or publicly available historic records, the informational advantage of such systems in knowing the actual realization gives the possibility of inducing a range of traffic flows through appropriate route recommendation or public information strategies. 

A strategy of a recommendation system to map state realization to randomized private route recommendations for the agents is referred to as a \emph{private signaling policy}; a strategy to map state realization to randomized public messages is referred to as a \emph{public signaling policy}. \revisionchange{The implementation of a private signaling policy requires the ability to provide different route recommendation to different agents. This can be achieved through personal mobile devices. On the other hand, public signaling policies require broadcasting the same message to all the agents. This can be achieved though road side variable message signs or through personal mobile devices. If the state corresponds to \texttt{incident} or \texttt{no incident}, then, e.g., the message space for the public policy can be the same, with no broadcast when the message generated by the policy is \texttt{no incident}. Alternately, a message could also be a route recommendation.}
A private signaling policy is feasible or \emph{obedient}, if, to every agent, it recommends a route which is weakly better in expectation, with respect to the induced posterior, than the other routes. Under a public signaling policy, the agents can be assumed to choose routes consistent with Bayes Nash flow with respect to the posterior. 
The problem of minimizing expected social latency cost over all obedient private or over all public signaling policies is referred to as \emph{information design} in this paper. We are interested in these problems for \emph{non-atomic} agents, when a fraction of agents do not participate in signaling and induce Bayes Nash flow with respect to the prior. The technical challenge is the joint consideration of optimal signaling policy for participating agents and the flow induced by non-participating agents. 
 
Information design for \emph{finite} agents has attracted considerable attention recently with applications in multiple domains, e.g., see \cite{Bergemann.Morris:19} for an overview; the single agent case was studied in \cite{Kamenica.Gentzkow:11} as \emph{Bayesian persuasion}. In the finite agent (and finite action) setting, the obedience condition on the signaling policy can be expressed as finite linear constraints, one for each combination of actions by the agents. This allows to cast the information design problem as a tractable optimization problem. Techniques to further reduce computational cost of information design are presented in \cite{Dughmi.Xu:16}. However, analogous computational approaches to solve information design for non-atomic agents, particularly for routing games, are lacking. 

%

There has been a growing interest recently in understanding the impact of information in non-atomic routing games. For example, \cite{Acemoglu.Makhdoumi.ea:16,Vasserman.Feldman.ea:15,Das.Kamenica.ea:17,Wu.Amin:19} illustrate that revealing full information to all the agents may not minimize social cost.
Information \emph{design} using private signaling policies, as in this paper, has also been pursued recently in \cite{Tavafoghi.Teneketzis:20}. Optimal public signaling policies for some settings were characterized in \cite{Massicot.Langbort:19}.  
While these works provide useful insights, the information \emph{design} aspect of these works is restricted to stylized settings involving a network with just two parallel links, sub-optimal policies, and link latency functions which ensure non-zero flow on all links under all state realizations. It is not apparent to what extent can the methodologies underlying these studies, which typically rely on analytical solutions, be generalized. Motivated by this, we develop a computational approach in this paper. While the detailed discussion is presented for parallel networks for simplicity in presentation, we also describe the extension of the computational framework to general networks with a single origin-destination pair.

Our key observation is that information design for polynomial latency functions has strong connections with the \emph{generalized problem of moments} (GPM)~\cite{Lasserre:08}. A GPM minimizes, over finite probability measures, a cost which is linear in moments with respect to these measures subject to constraints which are also linear in the moments. This connection allows to leverage computational tools developed for GPM, such as \texttt{GloptiPoly}~\cite{gloptipoly3}, which utilizes a hierarchy of semidefinite relaxations to lower bound GPM arbitrarily closely.  
For a fixed flow induced by non-participating agents, information design for participating agents is indeed a GPM. Furthermore, since the cost and constraints involve moments up to a finite order, there exists an optimal signaling policy which is atomic with provable upper bound on the number of atoms~\cite{Bayer.Teichmann:06}. \revisionchange{In other words, interestingly, a
finite-support, atomic signaling policy can achieve optimal performance.} 
This property also allows to \emph{equivalently} cast information design, when the non-participating agents choose route according to Bayes Nash flow, as a polynomial optimization problem. This can be arbitrarily lower bounded by hierarchy of semidefinite relaxations~\cite{Lasserre:01}, which can also be implemented in \texttt{GloptiPoly}. 
The first level of this hierarchy is shown to be exact for the basic two link case with affine latency functions, and it relies on using convexity of cost function and constraints to sharpen the bound from \cite{Bayer.Teichmann:06} for optimal solution. \revisionchange{The lower bound obtained from the hierarchy can be used to upper bound the optimality gap of a feasible solution obtained by packages such as \texttt{MultiStart} in \texttt{Matlab}. Indeed, in our simulations, we report number of starting points for \texttt{MultiStart} and the relaxation order for \texttt{GloptiPoly} for which this gap was found to be negligible.}

It is natural to compare our approach with semidefinite programming based approaches for computation of (Bayes) correlated equilibria, e.g., in \emph{continuous} polynomial games~\cite{Stein.Parrilo.ea:GEB11}. In \cite{Stein.Parrilo.ea:GEB11}, the action set is continuous and the agents are finite, and hence alternate formulations for correlated equilibrium are proposed which involve approximation through finite moments and discretization of the action set. In our setup, where the action set is finite and agents are non-atomic, the constraints for participating agents are readily in computational form and involve moments up to a finite order without any approximation. This then allows us to consider an \emph{equivalent} finite discretization, with known cardinality, of the agent population, to transform equivalently into polynomial optimization. Thereafter, the use of semidefinite relaxation hierarchy is standard. 

The computational approach of this paper can be utilized to complement the current studies on (paradoxical) effect of different fractions of participating agents under specific public signaling policies (primarily, full information). While existing works, e.g., \cite{Mahmassani.Jayakrishnan:91,Wu.Amin.ea:18}, study the effect on population-specific (i.e., participating and non-participating) costs, we study the effect on the social cost,  in the spirit of the social planner's perspective adopted in the paper. 
We provide a class of private signaling policies under which the optimal social cost is non-increasing with increasing fraction of participating agents. The key idea is to use an optimal solution at a given fraction to synthesize signaling policies which are feasible for all higher fractions and give the same cost. This monotonic result does not require the link latency functions to be polynomial. \revisionchange{On the other hand, we illustrate through examples that public signaling policies may worsen social performance if too many agents receive the signal.}

In summary, the main contributions of the paper are as follows. First, by making connection to GPM and associated semidefinite programming machinery, we point to a compelling computational framework to solve information design problems. Second, by establishing the existence of an atomic optimal solution, we provide credence to such a structural assumption often implicitly made in information design studies. The sharpening of the bound on the number of atoms that we illustrate in a simple case suggests the possibility of using the problem structure to reduce the size of the optimization formulation, and hence the computation cost. 
Third, the result and underlying proof technique for the monotonic behavior of social cost under a reasonable class of private signaling policies imply that private signaling policies can guarantee performance which is robust to higher than anticipated agent participation rate. However, our results also suggest that this may be difficult to achieve through public signaling policies. 
%
Overall, the contributions allow to considerably expand the scope of information design studies which has been limited so far to stylized settings. 

The rest of the paper is organized as follows. 
\ifthenelse{\equal{\isarxiv}{1}}
{Section~\ref{sec:problem} formulates the information design problem for non-atomic routing game on parallel networks,  and describes indirect signaling policies, with focus on private and public policies.}{\revisionchange{Section~\ref{sec:problem} formulates the information design problem for non-atomic routing game,  and describes private and public signaling policies.}} Section~\ref{sec:sdp} describes an exact polynomial optimization framework, and a class of signaling policies over which the social cost is non-increasing with increasing fraction of participating agents. \revisionchange{Section~\ref{sec:extension} provides extension to non-parallel networks and discusses computational complexity of the framework suggested in the paper to solve information design.} 
Section~\ref{sec:simulations} provides illustrative simulation results, and concluding remarks are provided in Section~\ref{sec:conclusions}. The proofs for all the technical results are provided in the Appendix.


\textbf{Notations}: $\triangle(X)$ denotes the set of all probability distributions on $X$. 
For an integer \( n \), we let \( [n]:= \{1, 2, \ldots, n\} \). 
For a vector $x \in \real^n$, let $\supp(x):=\setdef{i \in [n]}{x_i \neq 0}$ be the set of indices whose corresponding entries in $x$ are not zero. 
For $\lambda \geq 0$, let $\simplex_n(\lambda):=\setdef{x \in \real_{\geq 0}^n}{\sum_{i \in [n]} x_i = \lambda}$ be the $(n-1)$-dimensional \revisionchange{simplex} of size $\lambda$.
When $\lambda=1$, we shall simply denote the simplex as $\simplex_n$ for brevity in notation. 
$\zerobf_{n \times m}$ and $\onebf_{n \times m}$ will denote $n \times m$ matrices all of whose entries are  $0$ and $1$ respectively. 
In all these notations, the subscripts corresponding to size shall be omitted when clear from the context. For matrices $A$ and $B$ of the same size, their \revisionchange{Frobenius} inner product is $A \cdot B = \sum_{i, j} A_{i,j} B_{i,j}$.

\section{Problem Formulation and Preliminaries}
\label{sec:problem}
Consider a network consisting of $\nlinks$ parallel links between a single source-destination pair. \footnote{\revisionchange{Extension to non-parallel networks is discussed in Section~\ref{sec:extension}, where $\nlinks$ denotes the number of \emph{routes} between the origin-destination pair.}} \revisionchange{We use \emph{link} and \emph{route} interchangeably for parallel network.} 
Without loss of generality, let the agent population generate a unit volume of traffic demand. 
The link latency functions $\congfunc_{\state,i}(f_i)$, $i \in [\nlinks]$, give latency on link $i$ as a function of flow $f_i$ through them, conditional on the \emph{state} of the network $\state \in \allstates=\{\state_1, \ldots,\state_{\nstates}\}$. Throughout the paper, we shall make the following basic assumption on these functions.
\begin{assumption}
\label{ass:continuous}
For every $i \in [\nlinks]$, $\revisionchange{\state \in \allstates}$, $\congfunc_{\state,i}$ is a non-negative, continuously differentiable and non-decreasing function.
\end{assumption}
At times, we shall strengthen the assumption to $\congfunc_{\state,i}$ being strictly increasing. A class of functions satisfying Assumption~\ref{ass:continuous} which is attractive from a computational perspective is that of polynomial functions:
\begin{equation}
\label{eq:latency-affine}
\congfunc_{\state,i}(f_i) = \sum_{d=0}^D \alpha_{d,\state,i} \, f^d_i, \qquad i \in [\nlinks], \, \, \revisionchange{\state \in \allstates} 
\end{equation}
with $\alpha_{0,\state,i} \geq 0$ and $\alpha_{1,\state,i} \geq 0$. We shall also let $\alpha_d$ refer to the $\nstates \times \nlinks$ matrix whose entries are $\alpha_{d,\state,i}$.
Two instances of \eqref{eq:latency-affine} commonly studied in the literature are affine and the Bureau of Public Roads (BPR) functions~\cite{Branston:76}. In the former case, $D=1$ and in the latter case, $D=4$ with $\alpha_1=\alpha_2=\alpha_3=\zerobf$. Furthermore, the BPR function has the following interpretation: $\alpha_{0,\state,i}$ is the free flow travel time on link $i$ when the state is $\state$, and $\left(0.15 \,\frac{ \alpha_{0,\state,i}}{\alpha_{4,\state,i}}\right)^{\frac{1}{4}}$ is the flow capacity of link $i$ when the state is $\state$. 

Let $\state \sim \prior \in \texttt{interior}(\triangle(\allstates))$, for some prior $\prior$ which is known to all the agents. 
The agents do not have access to the realization of $\state$, but a fixed fraction $\pfrac \in [0,1]$ of the agents receives \revisionchange{private route recommendations or public messages conditional on the realized state.} 

\subsection{Private Signaling Policies}
\label{sec:direct}
The conditional route recommendations are generated by a \emph{private signaling policy} $\signal=\{\signal_{\state} \in \triangle(\simplex_{\nlinks}(\pfrac)): \, \state \in \allstates\}$ as follows. Given a realization $\state \in \allstates$, sample a $x \in \simplex_{\npaths}(\pfrac)$ according to $\signal_{\state}$, and partition the agent population into
 $\npaths+1$ parts with volumes $(x_1, \ldots, x_{\npaths},1-\pfrac)$. 
 All the agents are identical, and therefore in the non-atomic setting that we are considering here the partition can be formed by independently assigning every agent to a partition with probability equal to the volume of that partition. The agents in the $(\npaths+1)$-th partition, with volume $1-\pfrac$, do not receive any recommendation, whereas all the agents in the $i$-th partition, $i \in [\npaths]$, receive recommendation to take route $i$, \revisionchange{i.e., $x_i$ volume of agents is recommended to take route $i$.}

\revisionchange{
\begin{example}
Let $\Omega=\{\omega_1,\omega_2\}$. 
An example of a signaling policy for the two-link case with $\pfrac=1$ is: $\signal_{\state_1}=x_1$ with probability $0.5$ and $=x_2$ with probability $0.5$, $\signal_{\state_2}=x_1$ with probability $0.25$ and $=x_2$ with probability $0.75$, with $x_1=[0.75 \quad 0.25]^T$ and $x_2=[0.5 \quad 0.5]^T$. 
When the state is $\state_1$, then the planner recommends route 1 to $0.75$ volume of agents and route 2 to the remaining $0.25$ volume with probability $0.5$, and recommends route 1 to $0.5$ volume of agents and route 2 to the remaining $0.5$ volume with probability $0.5$. $\signal_{\state_2}$ has a similar interpretation. 

The special case when $\signal_{\state_1}$ and $\signal_{\state_2}$ are probability mass functions, as in this example, will be later referred to as \emph{atomic} private signaling policies and will play an important role in the paper. 
\end{example}
}

The policy $\signal$ and the fraction $\pfrac$ is publicly known to all the agents. Therefore, 
it is easy to see that the (joint) posterior on $(x,\state)$, i.e., the proportion of agents getting different recommendations and the state of the network, formed by an agent who receives recommendation $i \in [\npaths]$ is: 
\begin{equation}
\label{eq:posterior-expr}
\posterior^{\signal,i}(x,\state) = \frac{x_i \, \signal_{\state}(x) \, \prior(\state)}{\sum_{\theta \in \allstates} \int_{p \in \simplex(\pfrac)} p_i \, \signal_{\theta}(p) \, \de p \, \prior(\theta)}
\end{equation}
and the posterior formed by an agent who does not receive a recommendation is:
\begin{equation}
\label{eq:B-agent-joint-posterior}
\posterior^{\signal,\emptyset}(x,\state) = \signal_{\state}(x) \prior(\state) 
\end{equation}
\begin{remark}
One could consider an alternate setup where the set of agents who do not participate in the signaling scheme is pre-determined. These agents do not receive a recommendation and also do not have knowledge about $\signal$. In this case, \eqref{eq:B-agent-joint-posterior} can be replaced with $\posterior^{\signal,\emptyset}(x,\state)=\frac{\prior(\state)}{|\simplex(\pfrac)|}$ obtained by replacing $\signal_{\state}$ with the uniform distribution. The methodologies developed in this paper also extend to this alternate setting.
\end{remark}

A signaling policy is said to be \emph{obedient} if the recommendation received by every agent is weakly better, in expectation with respect to posterior in \eqref{eq:posterior-expr}, than other routes, while the non-participating agents induce a Bayes Nash flow with respect to their posterior in \eqref{eq:B-agent-joint-posterior}.  
Formally, a $\signal$ is said to be obedient\footnote{\revisionchange{An obedient signaling policy can be interpreted as a Bayes correlated equilibrium~\cite{Bergemann.Morris:16}.}} if there exists $y \in \simplex_{\npaths}(1-\pfrac)$ such that\footnote{Throughout the paper, unless noted otherwise, the summation over indices for degree, state and link, such as $d$, $\state$ and $i$, respectively, are to be taken over the entire range, i.e., $\{0, \ldots, D\}$, \revisionchange{$\allstates$} and $[\nlinks]$, respectively, \revisionchange{and the integral w.r.t. $x$ is over $\simplex(\pfrac)$}.}:
\ifthenelse{\equal{\numofcolumns}{1}}
{
\begin{subequations}
\label{eq:obedience-hetero-v1}
\begin{align}
\sum_{\state} \int_x \, \ell_{\state,i}(x_i + y_i) \posterior^{\signal,i}(x,\state) \, \de x & \leq \sum_{\state} \int_x \, \ell_{\state,j}(x_j + y_j) \posterior^{\signal,i}(x,\state) \, \de x, \quad i, j \in [\npaths] 
\label{eq:obedience-hetero-v1:obedience} \\
\sum_{\state} \int_x\, \ell_{\state,i}(x_i + y_i) \posterior^{\signal,\emptyset}(x,\state)\, \de x &
\leq \sum_{\state} \int_x\, \ell_{\state,j}(x_j + y_j) \posterior^{\signal,\emptyset}(x,\state)\, \de x, \quad i \in \supp(y), \, j \in [\npaths] 
\label{eq:obedience-hetero-v1:nash}
\end{align}
\end{subequations}
}{
\begin{subequations}
\label{eq:obedience-hetero-v1}
\begin{align}
\sum_{\state} & \int_x  \, \ell_{\state,i}(x_i + y_i) \posterior^{\signal,i}(x,\state) \, \de x \nonumber \\ & \leq \sum_{\state} \int_x \, \ell_{\state,j}(x_j + y_j) \posterior^{\signal,i}(x,\state) \, \de x, \, i, j \in [\npaths] \label{eq:obedience-hetero-v1:obedience}\\
\sum_{\state} & \int_x  \, \ell_{\state,i}(x_i + y_i) \posterior^{\signal,\emptyset}(x,\state)\, \de x \nonumber \\ &
\leq \sum_{\state} \int_x\, \ell_{\state,j}(x_j + y_j) \posterior^{\signal,\emptyset}(x,\state)\, \de x, \, i \in \supp(y), \, j \in [\npaths] 
\label{eq:obedience-hetero-v1:nash}
\end{align}
\end{subequations}
}
\revisionchange{$y$ is the flow induced by the non-participating agents. \eqref{eq:obedience-hetero-v1:nash} captures the fact that this is the Bayes Nash flow with respect to the prior.}
Plugging the expressions of beliefs from \eqref{eq:posterior-expr} and \eqref{eq:B-agent-joint-posterior}, noting that the denominators on both sides of the inequalities are the same in \eqref{eq:obedience-hetero-v1}, and multiplying both sides of the second set of inequalities by $y_i$, one equivalently gets:
\ifthenelse{\equal{\numofcolumns}{1}}
{
\begin{subequations}
\label{obedience-hetero-v2:main}
\begin{align}
\sum_{\state} \int_x \, \ell_{\state,i}(x_i + y_i) \, x_i \, \signal_{\state}(x) \de x \, \prior(\state) & \leq \sum_{\state} \int_x \, \ell_{\state,j}(x_j + y_j) \, x_i \, \signal_{\state}(x) \de x\, \prior(\state), \quad i, j \in [\npaths] \label{obedience-hetero-v2:obed}\\
\sum_{\state} \int_x\, \ell_{\state,i}(x_i + y_i) \, y_i \, \signal_{\state}(x) \de x \, \prior(\state) &
\leq \sum_{\state} \int_x\, \ell_{\state,j}(x_j + y_j) \, y_i \, \signal_{\state}(x) \de x \, \prior(\state), \quad i, j \in [\npaths] \label{obedience-hetero-v2:nash}
\end{align}
\end{subequations}
}{
\begin{subequations}
\label{obedience-hetero-v2:main}
\begin{align}
\sum_{\state} & \int_x \, \ell_{\state,i}(x_i + y_i) \, x_i \, \signal_{\state}(x) \de x \, \prior(\state) \nonumber \\ & \leq \sum_{\state} \int_x \, \ell_{\state,j}(x_j + y_j) \, x_i \, \signal_{\state}(x) \de x\, \prior(\state), \quad i, j \in [\npaths] \label{obedience-hetero-v2:obed}\\
\sum_{\state} & \int_x\, \ell_{\state,i}(x_i + y_i) \, y_i \, \signal_{\state}(x) \de x \, \prior(\state) \nonumber \\ &
\leq \sum_{\state} \int_x\, \ell_{\state,j}(x_j + y_j) \, y_i \, \signal_{\state}(x) \de x \, \prior(\state), \quad i, j \in [\npaths] \label{obedience-hetero-v2:nash}
\end{align}
\end{subequations}
}
We emphasize that multiplying both sides by $y_i$ allows to equivalently relax the restriction on $i$ in terms of $y$ in \eqref{eq:obedience-hetero-v1:nash} to get \eqref{obedience-hetero-v2:nash}.

The social cost is taken to be the expected total latency:
\begin{equation}
\label{info-design:main}
J(\signal,y):=  \sum_{\state, \, i} \int_x  \left(x_i + y_i\right) \, \ell_{\state,i} (x_i + y_i) \, \signal_{\state}(x) \de x  \, \prior(\state)
\end{equation}
The information design problem can then be stated as 
\begin{equation}
\label{eq:info-design-joint}
\min_{(\signal, y) \in \allsignals \times \simplex(1-\pfrac)} \, \, \, J(\signal,y) \, \, \text{s.t. } \eqref{obedience-hetero-v2:main}
\end{equation}
where $\allsignals$ is the concise notation for ${\triangle(\simplex(\pfrac))}^{\nstates}$.

\begin{remark}
\begin{enumerate}
\item[(i)] If there are multiple feasible $y$ for a given $\signal$, then a solution $(\signal^*,y^*)$ to \eqref{eq:info-design-joint} can be interpreted as implicitly requiring an additional effort from the social planner to enforce $y^*$. One could alternately consider a \emph{robust} formulation by replacing  $\min_{(\signal, y)}$ in \eqref{eq:info-design-joint} with $\min_{\signal} \, \max_{y}$. We leave such an extension for future consideration. Moreover, as we state below after the remark, under a rather reasonable condition on the link latency functions, there exists a unique feasible $y$ for every $\signal$, in which case the robust version is the same as \eqref{eq:info-design-joint}.
\ifthenelse{\equal{\isarxiv}{1}}
{
}{
\item[(ii)] \revisionchange{One can show that revelation principle, e.g., see \cite{Bergemann.Morris:19}, holds true in the setting of this paper for strictly increasing link latency functions~\cite[Section IIB]{Zhu.Savla:infodesign-arxiv20}.} 
This implies that optimality in the class of obedient \emph{direct} private signaling policies, i.e., signaling policies which recommend routes, also ensures optimality within a broader class which includes \emph{indirect} signaling policies. An indirect signaling policy provides noisy information about the state realization. The route choice is then determined by Bayes Nash flow with respect to the posterior beliefs induced by the signaling policy. In Section~\ref{sec:pub-signals}, we consider a special case of indirect policies, known as \emph{public} signaling policies.
}
\ifthenelse{\equal{\isarxiv}{1}}
{
\item[(ii)] The feasible set in \eqref{eq:info-design-joint} is non-empty for all $\pfrac \in [0,1]$. Details are provided in Remark~\ref{rem:public-private-connection}.
}{
\item[(iii)] The feasible set in \eqref{eq:info-design-joint} is non-empty for all $\pfrac \in [0,1]$. Details are provided in Remark~\ref{rem:public-private-connection}.
}
\end{enumerate}
\end{remark}

It can be shown easily using a straightforward adaptation of the standard argument for Wardrop equilibrium in the deterministic case that, for every $\signal \in \allsignals$, a $y \in \simplex(1-\pfrac)$ satisfies \eqref{obedience-hetero-v2:nash} if and only if it solves the following convex problem:
\begin{equation}
\label{eq:nash-potential}
\min_{y \in \simplex(1-\pfrac)} \, \sum_{\state, \, i}  \, \int_0^{y_i} \, \int_{x} \, \congfunc_{\state, i} (x_i + s) \, \signal_{\state}(x) \de x\, \de s \, \prior(\state)
\end{equation}
Moreover, such a $y$ is unique if $\{\congfunc_{\state,i}\}_{\state,i}$ are strictly increasing over $[0,1]$. In particular, for uniqueness, it is sufficient to have $\alpha_{1,\state,i}>0$ for all $\state, i$ for affine latency functions, and $\alpha_{4,\state,i}>0$ for all $\state, i$ for BPR latency functions.\footnote{Note that all the derivatives of the BPR latency function are zero at $0$. However, one can easily show uniqueness in the special cases when, for a signaling policy supported only on $x_i=0$, \eqref{eq:nash-potential} has a solution with $y_i=0$.}

\ifthenelse{\equal{\isarxiv}{1}}
{
\subsection{Indirect Signaling Policies}
\label{sec:indirect}
The private signaling policies considered in the previous section are \emph{direct}, i.e., their output space is the set of routes. A generalization is when the output space is arbitrary set of messages, e.g., travel time on the routes. Let the message space be $\{1,\ldots,\nmesgs\}=[\nmesgs]$.  Formally, an \emph{indirect signaling policy} is $\indsignal=\{\indsignal_{\state} \in \triangle(\simplex_{\nmesgs}(\pfrac)): \, \state \in \allstates\}$. The policy generates a message vector $\bar{x} \in \simplex_{\nmesgs}(\pfrac)$, where $\bar{x}_h$ is the volume of agents who get message $k \in [\nmesgs]$. The joint posterior formed by an agent who receives message $k$ is:
\begin{equation*}
\label{eq:posterior-expr-indirect}
\posterior^{\indsignal,k}(\bar{x},\state) = \frac{\bar{x}_k \, \indsignal_{\state}(\bar{x}) \, \prior(\state)}{\sum_{\theta \in \allstates} \int_{p \in \simplex_{\nmesgs}(\pfrac)} p_h \, \indsignal_{\theta}(p) \, \de p \, \prior(\theta)}
\end{equation*}
and the posterior formed by an agent who does not receive a recommendation is:
\begin{equation*}
\label{eq:B-agent-joint-posterior-indirect}
\posterior^{\indsignal,\emptyset}(\bar{x},\state) = \indsignal_{\state}(\bar{x}) \prior(\state) 
\end{equation*}

%

Let $x^{(k)} \in \simplex(\pfrac)$ be the link flow induced by the agents receiving message $k$, and let $y \in \simplex(1-\pfrac)$ be the link flow induced by agents not receiving the message. These link flows are given by the \emph{Bayes Nash equilibrium} (BNE) of the underlying Bayesian game, i.e., they satisfy: $\forall i, j \in [\nlinks]$,
\begin{equation}
\label{eq:BNE}
\begin{split}
\sum_{\state} \int_{\bar{x}} \bar{x}_k \, x^{(k)}_{i} \congfunc^{\state}_i\big(\sum_r \bar{x}_{r} \, x^{(r)}_{i} +y_i\big) \indsignal_{\state}(\bar{x}) \prior(\state) \de \bar{x}& \leq \sum_{\state} \bar{x}_{k} \, x^{(k)} _{i} \congfunc^{\state}_j(\sum_r \bar{x}_{r} \, x^{(r)}_{j} +y_j) \indsignal_{\state}(\bar{x}) \prior(\state) \de \bar{x}, \quad \forall  k \in [\nmesgs]
\\
\sum_{\state} \int_{\bar{x}} y_i \, \congfunc^{\state}_i\big(\sum_r \bar{x}_{r} \, x^{(r)}_{i}+y_i\big) \indsignal_{\state}(\bar{x}) \prior(\state) \de \bar{x} & \leq \sum_{\state} \int_{\bar{x}} y_i \, \congfunc^{\state}_j\big(\sum_r \bar{x}_r \, x^{(r)}_{j}+y_j\big) \indsignal_{\state}(\bar{x}) \prior(\state) \de \bar{x}
\end{split}
\end{equation}

We next discuss existence and equivalence of BNE link flows.

\begin{proposition}
\label{prop:revelation}
$(x,y) \equiv (\{x^{(k)}: \, k \in [\nmesgs]\}, y)$ is a BNE flow for an indirect signaling policy $\indsignal$ if and only if it is a solution to:
\begin{equation}
\label{eq:BNE-flow-solution}
\min_{\substack{y \in \simplex(1-\pfrac) \\ x^{(k)} \in \simplex(\pfrac), \, k \in [\nmesgs]}} \quad \sum_{i, \, \state} \, \, \int_{\bar{x}} \, \, \int_0^{\sum_k \bar{x}_k \, x^{(k)}_{i} + y_i} \, \congfunc_{\state,i}(z)  \, \indsignal_{\state}(\bar{x}) \, \prior(\state) \, \de z \, \de \bar{x}
\end{equation}

Furthermore, if the link latency functions are strictly increasing, then all the BNE flows associated with a policy have the same aggregate link flow, i.e., for any two BNE flows $(x^{(1)},y^{(1)})$ and $(x^{(2)},y^{(2)})$, we have $\sum_k \bar{x}_k \, x^{(1,k)}+y^{(1)}=\sum_k \bar{x}_k \,x^{(2,k)}+y^{(2)}$ for all $\bar{x} \in \simplex_{\nmesgs}(\pfrac)$.
\end{proposition}

\begin{remark}
\begin{enumerate}
\item 
Direct private signaling policies in Section~\ref{sec:direct} correspond to the special case of indirect policies when the set of messages is equal to the set of routes. Accordingly obedience condition in \eqref{obedience-hetero-v2:main} is derived from \eqref{eq:BNE} with $x^{(k)}_i=\pfrac$ if $i=k$ and equal to zero otherwise. 
\item Proposition~\ref{prop:revelation} implies that the revelation principle, e.g., see \cite{Bergemann.Morris:19}, holds true in the setting of this paper. That is, for every indirect policy, there exists a direct policy which induces the same aggregate link flows, and therefore, it is sufficient to optimize over the class of direct policies.  
\end{enumerate}
\end{remark}

}

\subsection{Public Signaling Policies}
\label{sec:pub-signals}
\ifthenelse{\equal{\isarxiv}{1}}
{
Public signaling policies are the special case of indirect signaling policies when $\indsignal_{\state}$ is supported on the $\nmesgs$ vertices of the simplex $\simplex_{\nmesgs}$. Therefore, a public signaling policy is a map $\map{\pubsignal}{\allstates}{\triangle([\nmesgs])}$, or can alternately be represented as a $\nstates \times \nmesgs$ row stochastic matrix. In the setting of public signaling policies, the BNE flow condition in \eqref{eq:BNE} becomes:
\begin{subequations}
\label{eq:BNE-public}
\begin{align}
x^{(k)}_i \sum_{\state} \left(\congfunc_{\state,i}(x^{(k)}_i+y_i) - \congfunc_{\state,j}(x^{(k)}_j+y_j) \right)\, \pubsignal(k|\state)\prior(\state) & \leq 0, \qquad i, j \in [\nlinks], \quad k \in [\nmesgs] \label{eq:BNE-public-receiving}
\\
y_i \sum_{k, \, \state} \left(\congfunc_{\state,i}(x^{(k)}_i+y_i) - \congfunc_{\state,j}(x^{(k)}_j+y_j) \right)\, \pubsignal(k|\state)\prior(\state) & \leq 0, \qquad i, j \in [\nlinks] \label{eq:BNE-public-nonreceiving}
\end{align}
\end{subequations}
%
}{
A public signaling policy is an indirect signaling policy, under which, for every state realization, $\pfrac$ fraction of agents all receive the same message among $\{1,\ldots,\nmesgs\}=[\nmesgs]$. 
Formally, a public signaling policy is a map $\map{\pubsignal}{\allstates}{\triangle([\nmesgs])}$, or can alternately be represented as a $\nstates \times \nmesgs$ row stochastic matrix. 
The posterior formed by agents when the message they receive is $k$ is:
\begin{equation}
\label{eq:posterior-receiving}
\posterior^{\pubsignal,k}(\state)= \frac{\pubsignal(k|\state)\prior(\state)}{\sum_{\theta} \pubsignal(k|\theta)\prior(\theta)}, \qquad \state \in \allstates
\end{equation}
The joint posterior formed by agents who do not receive message, but have knowledge of $\pubsignal$, is:
\begin{equation}
\label{eq:posterior-not-receiving}
\posterior^{\pubsignal,\emptyset}(k,\state)=\pubsignal(k|\state)\prior(\state), \qquad k \in [\nmesgs], \, \state \in \allstates
\end{equation}

Let $x^{(k)} \in \simplex(\pfrac)$ be the link flow induced by participating agents, when the message they receive is $k \in [\nmesgs]$, and let $y \in \simplex(1-\pfrac)$ be the link flow induced by agents not receiving the message. $x^{(k)}$ is the Bayes Nash flow with respect to the posterior in \eqref{eq:posterior-receiving} and $y$ is the Bayes Nash flow with respect to the posterior in \eqref{eq:posterior-not-receiving}. That is, $x^{(k)}$ satisfies:  
\ifthenelse{\equal{\numofcolumns}{1}}
{
\begin{equation*}
\sum_{\state} \, \congfunc_{\state,i}(x^{(k)}_i+y_i) \posterior^{\pubsignal,k}(\state) \leq \sum_{\state} \, \congfunc_{\state,j}(x^{(k)}_j+y_j) \posterior^{\pubsignal,k}(\state), \qquad i \in \supp(x^{(k)}), \, \, j \in [\nlinks]
\end{equation*}
}{
\begin{multline*}
\sum_{\state} \, \congfunc_{\state,i}(x^{(k)}_i+y_i) \posterior^{\pubsignal,k}(\state) \leq \sum_{\state} \, \congfunc_{\state,j}(x^{(k)}_j+y_j) \posterior^{\pubsignal,k}(\state) \\ \qquad i \in \supp(x^{(k)}), \, \, j \in [\nlinks]
\end{multline*}
}
Substituting the expression from \eqref{eq:posterior-receiving}, the conditions on $\{x^{(1)}, \ldots, x^{(\nmesgs)}\}$ can be collectively rewritten as 
\ifthenelse{\equal{\numofcolumns}{1}}
{
\begin{equation}
\label{eq:nash-receiving}
x^{(k)}_i \sum_{\state} \left(\congfunc_{\state,i}(x^{(k)}_i+y_i) - \congfunc_{\state,j}(x^{(k)}_j+y_j) \right)\, \pubsignal(k|\state)\prior(\state) \leq 0, \qquad i, j \in [\nlinks], \quad k \in [\nmesgs]
\end{equation}
}{
\begin{equation}
\label{eq:nash-receiving}
\begin{split}
x^{(k)}_i  \sum_{\state} & \left(\congfunc_{\state,i}(x^{(k)}_i+y_i) - \congfunc_{\state,j}(x^{(k)}_j+y_j) \right)\, \pubsignal(k|\state)\prior(\state) \\ & \leq 0, \qquad i, j \in [\nlinks], \quad k \in [\nmesgs]
\end{split}
\end{equation}
}
Similarly, the condition on $y$ can be written as
\ifthenelse{\equal{\numofcolumns}{1}}
{
\begin{equation}
\label{eq:nash-not-receiving}
y_i \sum_{k, \, \state} \left(\congfunc_{\state,i}(x^{(k)}_i+y_i) - \congfunc_{\state,j}(x^{(k)}_j+y_j) \right)\, \pubsignal(k|\state)\prior(\state) \leq 0, \qquad i, j \in [\nlinks]
\end{equation}
}{
\begin{equation}
\label{eq:nash-not-receiving}
\begin{split}
y_i \sum_{k, \, \state} & \left(\congfunc_{\state,i}(x^{(k)}_i+y_i) - \congfunc_{\state,j}(x^{(k)}_j+y_j) \right)\, \pubsignal(k|\state)\prior(\state) \\ & \leq 0, \qquad i, j \in [\nlinks]
\end{split}
\end{equation}
}
}

The social cost is:
\begin{equation}
\label{eq:social-cost-pub}
\begin{split}
J(\pubsignal,x,y) & := \sum_{k, \, i, \, \state} \, (x^{(k)}_i+y_i) \congfunc_{\state,i}(x^{(k)}_i+y_i) \pubsignal(k|\state)\prior(\state) 
\end{split}
\end{equation}
Therefore, the problem of optimal public signaling policy design can be written as:
\ifthenelse{\equal{\isarxiv}{1}}
{
\begin{equation}
\label{eq:info-design-pub}
\min_{\substack{x^{(k)} \in \simplex(\pfrac), \, k \in [\nmesgs] \\ y \in \simplex(1-\pfrac) \\ \pubsignal \in \allsignals(\nmesgs)}} \, J(\pubsignal,x,y) \quad \text{s.t. } \eqref{eq:BNE-public}
\end{equation}
}{
\begin{equation}
\label{eq:info-design-pub}
\min_{\substack{x^{(k)} \in \simplex(\pfrac), \, k \in [\nmesgs] \\ y \in \simplex(1-\pfrac) \\ \pubsignal \in \allsignals(\nmesgs)}} \, J(\pubsignal,x,y) \quad \text{s.t. } \eqref{eq:nash-receiving}-\eqref{eq:nash-not-receiving}
\end{equation}
}

\begin{example}
Two public signaling policies which have attracted particular interest are \emph{full information} and \emph{no information}:
\ifthenelse{\equal{\numofcolumns}{1}}
{}{
\vspace{-0.1in}
}

\ifthenelse{\equal{\numofcolumns}{1}}
{
\begin{equation}
\label{eq:pub-signal-eg}
\signal^{\text{pub, full}}=
\kbordermatrix{&k=1&k=2 & \ldots & k = \nmesgs\\
\state_1&1&0&\ldots&0\\
\state_2&0&1&\ldots&0\\
\vdots & \vdots & \vdots & \ldots & \vdots\\
\state_{\nstates} & 0 & 0 & \ldots & 1
}, \qquad
\signal^{\text{pub, no}}=
\kbordermatrix{&k=1&k=2 & \ldots & k = \nmesgs\\
\state_1&1&0&\ldots&0\\
\state_2&1&0&\ldots&0\\
\vdots & \vdots & \vdots & \ldots & \vdots\\
\state_{\nstates} & 1 & 0 & \ldots & 0
}
\end{equation}
}{
\begin{equation}
\label{eq:pub-signal-eg}
\begin{split}
\signal^{\text{pub, full}} & =
\kbordermatrix{&k=1&k=2 & \ldots & k = \nmesgs\\
\state_1&1&0&\ldots&0\\
\state_2&0&1&\ldots&0\\
\vdots & \vdots & \vdots & \ldots & \vdots\\
\state_{\nstates} & 0 & 0 & \ldots & 1
} \\
\signal^{\text{pub, no}} & =
\kbordermatrix{&k=1&k=2 & \ldots & k = \nmesgs\\
\state_1&1&0&\ldots&0\\
\state_2&1&0&\ldots&0\\
\vdots & \vdots & \vdots & \ldots & \vdots\\
\state_{\nstates} & 1 & 0 & \ldots & 0
}
\end{split}
\end{equation}
}
where $\nmesgs=\nstates$ for the full information signaling policy, and $\nmesgs$ is arbitrary, e.g., $\nmesgs=1$, for the no information signaling policy. In fact, any row-stochastic $\signal^{\text{pub, no}}$ with identical rows is a no information signaling policy.
\end{example}
\ifthenelse{\equal{\isarxiv}{1}}
{
}
{
It is sometimes of interest to evaluate the cost of a given public signaling policy. The cost can be computed from the induced flows $x^{(k)}$, $k \in [\nmesgs]$, and $y$. \revisionchange{These flows are solution to a convex optimization problem~\cite[Proposition 1]{Zhu.Savla:infodesign-arxiv20}.}

}

%

\ifthenelse{\equal{\numofcolumns}{1}}
{}{
\vspace{-0.1in}
}

\section{An Exact Polynomial Optimization Formulation for Private Signaling Policies}
\label{sec:sdp}
In this section, unless stated otherwise, we assume that the link latency functions are polynomial, i.e., of the form in \eqref{eq:latency-affine}. \revisionchange{For such latency functions, designing optimal public signaling policy in \eqref{eq:info-design-pub} for a given message space is a polynomial optimization problem. For example, \eqref{eq:info-design-pub} is a third degree polynomial optimization problem for affine link latency functions.
This is however not the case for private policies in \eqref{eq:info-design-joint}. We now describe a procedure to \emph{equivalently} convert \eqref{eq:info-design-joint} into a polynomial optimization problem.}

Let us first consider minimizing $J(\signal,y)$ over $\signal$ satisfying \eqref{obedience-hetero-v2:obed}, for a fixed $y$. Note that, for $y=\zerobf$, this corresponds to the information design problem in the special case when $\pfrac=1$. Even in this special case, which has been studied previously in \cite{Das.Kamenica.ea:17,Tavafoghi.Teneketzis:20}, no comprehensive solution methodology exists.   

We start by rewriting the information design problem in terms of moments of the signaling policy $\signal$. 
Let $z$ be the vector of all monomials in $x_1, \ldots, x_{\nlinks}$ up to degree \revisionchange{$\frac{D+1}{2}$ if $D$ is odd, and $\frac{D}{2}+1$ if $D$ is even}, arranged in a lexicographical order. For example, for $D=3$, \ifthenelse{\equal{\numofcolumns}{1}}
{$z=[1, \, x_1, \ldots, x_{\nlinks}, \, x^2_1, \ldots, x_1 x_{\nlinks}, \, x_2 x_1, \ldots, x_2 x_{\nlinks}, \ldots, x_{\nlinks} x_1, \ldots, x_{\nlinks}^2]^T$.}{$z=[1, \, x_1, \ldots, x_{\nlinks}, \, x^2_1, \ldots, x_1 x_{\nlinks}, \, x_2 x_1, \ldots, x_2 x_{\nlinks}, \ldots, x_{\nlinks} x_1,$ $\ldots, x_{\nlinks}^2]^T$.} 
For a fixed $y$, \eqref{eq:info-design-joint} can then be written as:
\begin{subequations}
\label{info-design-gmp-hetero:main}
\begin{align}
\min_{\signal \in \allsignals} & \, \, \sum_{\state} \int \, C_{\state}(y) \cdot z z^T \, \signal_{\state}(x) \de x
 \label{info-design-gmp-hetero:cost}\\
\text{s.t.} \quad & \sum_{\state} \int \,  A_{\state}^{(i,j)}(y) \cdot z z^T \, \, \signal_{\state} (x) \de x\geq 0, \quad i, j \in [\nlinks] \label{info-design-gmp-hetero:obedience}
\\
& \sum_{\state} \int \,  B_{\state}^{(i,j)}(y) \cdot z z^T \, \, \signal_{\state} (x) \de x\geq 0, \quad i, j \in [\nlinks] \label{info-design-gmp-hetero:nash}
\end{align}
\end{subequations}

\ifthenelse{\equal{\numofcolumns}{1}}
{}{
\vspace{-0.5in}
}
for appropriate symmetric matrices $C_{\state}$, $A_{\state}^{(i,j)}$, and $B_{\state}^{(i,j)}$; expressions for these matrices in the special case when $D=1$ (i.e., affine link latency functions) are provided in \ifthenelse{\equal{\isarxiv}{1}}{Appendix~\ref{sec:matrix-expressions}}{the extended version~\cite{Zhu.Savla:infodesign-arxiv20}}. The cost in \eqref{info-design-gmp-hetero:cost} is the same as the cost in \eqref{info-design:main}, \eqref{info-design-gmp-hetero:obedience} corresponds to the obedience constraint in \eqref{obedience-hetero-v2:obed}, and \eqref{info-design-gmp-hetero:nash} corresponds to \eqref{obedience-hetero-v2:nash}. 

\eqref{info-design-gmp-hetero:main} is an instance of the \emph{generalized problem of moments} (GPM)~\cite{Lasserre:08}, which in turn can be solved numerically using \texttt{GloptiPoly}~\cite{gloptipoly3}. This software solves GPM by lower bounding it with semidefinite relaxations of increasing order. The stopping criterion on the order is however problem-dependent; approximations can be obtained by a user-specified order. In the special of $\nlinks=2$, the first order relaxation is tight. 

\begin{proposition}
\label{prop:gpm-sdp-exact}
Let $\nlinks=2$. For every $y \in \simplex(1-\pfrac)$, \eqref{info-design-gmp-hetero:main} is equivalent to a semidefinite program.
\end{proposition}
\begin{remark}
Proposition~\ref{prop:gpm-sdp-exact} implies that, in the case of two links, when all the agents are participating, i.e., $\pfrac=1$, computing optimal signaling policy is tractable for arbitrary polynomial latency functions. This is to be contrasted with existing work, e.g., \cite{Das.Kamenica.ea:17,Tavafoghi.Teneketzis:20}, where an optimal signaling policy is provided for such a setting only for certain affine link latency functions.  
\end{remark}

\subsection{Atomic Private Signaling Policies}
\label{sec:atomic}
A natural approach to approximate the joint optimization in \eqref{eq:info-design-joint} is to discretize the support of $\signal$. 
A signaling policy $\signal$ is called \emph{$\nmesgs$-atomic}, $\nmesgs \in \NN$, if, for every $\state \in \allstates$, $\signal_{\state}$ is supported on $\nmesgs$ discrete points $x^{(k)} \in \simplex(\pfrac)$, $k \in [\nmesgs]$. Let the set of such signaling policies be denoted as $\allsignals(\nmesgs)$. It is easy to see that every signaling policy in $\allsignals(\nmesgs)$ can be represented as a $\nstates \times \nmesgs$ row stochastic matrix. To emphasize the matrix notation, we let $\signal(k|\state)$ denote the probability of recommending routes according to $x^{(k)}$ when the state realization is $\state$. 
Computing optimal signaling policy in $\allsignals(m)$ can be written as the following polynomial optimization problem:
\footnote{Throughout the paper, unless noted otherwise, the summation over index for discrete support, such as $k$, is to be taken over the entire range, i.e., $\nmesgs$.}:
\ifthenelse{\equal{\numofcolumns}{1}}
{
\begin{subequations}
\label{info-design-atomic:main}
\begin{align}
\min_{\substack{x^{(k)} \in \simplex(\pfrac), \, k \in [\nmesgs] \\ y \in \simplex(1-\pfrac) \\ \signal \in \allsignals(\nmesgs)}} & \, \, \,  \sum_{k, \, \state, \, i}  \left(x^{(k)}_i + y_i\right) \, \congfunc_{\state,i} (x^{(k)}_i + y_i) \, \signal(k|\state) \, \prior(\state) \label{info-design-atomic:cost}\\
\text{s.t. } & \, \, \sum_{k, \, \state} \, \congfunc_{\state,i}(x^{(k)}_i + y_i) \, x^{(k)}_i \, \signal(k|\state) \, \prior(\state) \leq \sum_{k, \, \state} \, \congfunc_{\state,j}(x^{(k)}_j + y_j) \, x^{(k)}_i \,  \signal(k|\state) \, \prior(\state), \quad i, j \in [\nlinks] \label{info-design-atomic:obedience}\\
& \, \, \sum_{k, \, \state} \, \congfunc_{\state,i}(x^{(k)}_i + y_i) \, y_i \, \signal(k|\state) \, \prior(\state) \leq \sum_{k, \, \state} \, \congfunc_{\state,j}(x^{(k)}_j + y_j) \, y_i \,  \signal(k|\state) \, \prior(\state), \quad i, j \in [\nlinks] \label{info-design-atomic:nash}
\end{align}
\end{subequations}
}{
\begin{subequations}
\label{info-design-atomic:main}
\begin{align}
& \min_{\substack{x^{(k)} \in \simplex(\pfrac), \, k \in [\nmesgs] \\ y \in \simplex(1-\pfrac) \\ \signal \in \allsignals(\nmesgs)}} \sum_{k, \, \state, \, i}  \left(x^{(k)}_i + y_i\right) \, \congfunc_{\state,i} (x^{(k)}_i + y_i)  \signal(k|\state) \prior(\state) \label{info-design-atomic:cost}\\
& \text{s.t. }  \, \sum_{k, \, \state} \, \congfunc_{\state,i}(x^{(k)}_i + y_i) \, x^{(k)}_i \, \signal(k|\state) \, \prior(\state) \nonumber \\
& \leq  \sum_{k, \, \state} \, \congfunc_{\state,j}(x^{(k)}_j + y_j) \, x^{(k)}_i \,  \signal(k|\state) \, \prior(\state), \quad i, j \in [\nlinks] \label{info-design-atomic:obedience}\\
& \sum_{k, \, \state} \, \congfunc_{\state,i}(x^{(k)}_i + y_i) \, y_i \, \signal(k|\state) \, \prior(\state) 
\nonumber \\
& \leq  \sum_{k, \, \state} \, \congfunc_{\state,j}(x^{(k)}_j + y_j) \, y_i \,  \signal(k|\state) \, \prior(\state), \quad i, j \in [\nlinks] \label{info-design-atomic:nash}
\end{align}
\end{subequations}
}

In particular, for \eqref{eq:latency-affine} with $D=1$, i.e., affine link latency functions, the polynomials in the cost functions and the constraints are of degree 3. 

\eqref{info-design-atomic:main} can also be solved (approximately) using \texttt{GloptiPoly}. \eqref{info-design-atomic:main} gives an increasingly tighter upper bound to \eqref{eq:info-design-joint} with increasing $\nmesgs \in \NN$. While it is natural to expect the gap between \eqref{info-design-atomic:main} and \eqref{eq:info-design-joint} to go to zero as $m \to + \infty$, the gap in fact becomes zero for finite $m$.

\begin{theorem}
\label{prop:natoms-upper-bound}
\eqref{eq:info-design-joint} is equivalent to \eqref{info-design-atomic:main} for $\nmesgs \geq \nstates \, {D+\nlinks \choose D+1}$.
\end{theorem}

The upper bound in Theorem~\ref{prop:natoms-upper-bound} on the number of atoms required to realize an optimal signaling policy can be tightened in some cases, as we show in the next section. 

\subsection{Diagonal Atomic Private Signaling Policies}
An atomic policy which has attracted particular attention is when $\signal$ is the identity matrix of size $\nstates$. We shall refer to such a policy as a \emph{diagonal atomic signaling policy}, and denote its finite support as $x^{\state}$, $\state \in \allstates$. 
\revisionchange{These policies are among the simplest policies which do not reveal the true state. They simplify the process of route recommendation for the planner, and also reduce the complexity of the information design problem. Besides, as shown in Section~\ref{sec:monotone}, they are an important medium for showing monotonicity of cost with increasing fraction of participating agents. Moreover, as simulations in Section~\ref{sec:simulations} suggest, it might be sufficient to focus on them for optimal performance, however a formal study is left to future work.} 
The polynomial optimization problem in \eqref{info-design-atomic:main} for diagonal atomic policies simplifies to:
\ifthenelse{\equal{\numofcolumns}{1}}
{
\begin{subequations}
\label{eq:info-design-diagonal}
\begin{align}
\min_{\substack{x^{\state} \in \simplex(\pfrac), \, \state \in \allstates \\ y \in \simplex(1-\pfrac)}} & \, \, \,  \sum_{\state, \, i}  \left(x^{\state}_i + y_i\right) \, \congfunc_{\state,i} (x^{\state}_i + y_i) \, \prior(\state) 
\label{eq:info-design-diagonal:cost}\\
\text{s.t. } & \, \, \sum_{\state} \, \congfunc_{\state,i}(x^{\state}_i + y_i) \, x^{\state}_i \, \prior(\state) \leq \sum_{\state} \, \congfunc_{\state,j}(x^{\state}_j + y_j) \, x^{\state}_i \, \prior(\state), \quad i, j \in [\nlinks] 
\label{eq:info-design-diagonal:obedience}\\
& \, \, \sum_{\state} \, \congfunc_{\state,i}(x^{\state}_i + y_i) \, y_i \, \prior(\state) \leq \sum_{\state} \, \congfunc_{\state,j}(x^{\state}_j + y_j) \, y_i \, \prior(\state), \quad i, j \in [\nlinks]
\label{eq:info-design-diagonal:nash}
\end{align}
\end{subequations}
}{
\begin{subequations}
\label{eq:info-design-diagonal}
\begin{align}
\min_{\substack{x^{\state} \in \simplex(\pfrac), \, \state \in \allstates \\ y \in \simplex(1-\pfrac)}} & \, \, \,  \sum_{\state, \, i}  \left(x^{\state}_i + y_i\right) \, \congfunc_{\state,i} (x^{\state}_i + y_i) \, \prior(\state) 
\label{eq:info-design-diagonal:cost}\\
\text{s.t. } \, \, \sum_{\state} \, & \congfunc_{\state,i}(x^{\state}_i + y_i) \, x^{\state}_i \, \prior(\state) \nonumber \\
\leq & \sum_{\state} \, \congfunc_{\state,j}(x^{\state}_j + y_j) \, x^{\state}_i \, \prior(\state), \quad i, j \in [\nlinks] 
\label{eq:info-design-diagonal:obedience}\\
\sum_{\state} \, & \congfunc_{\state,i}(x^{\state}_i + y_i) \, y_i \, \prior(\state) 
\nonumber \\
\leq & \sum_{\state} \, \congfunc_{\state,j}(x^{\state}_j + y_j) \, y_i \, \prior(\state), \quad i, j \in [\nlinks]
\label{eq:info-design-diagonal:nash}
\end{align}
\end{subequations}
}
In general, \eqref{eq:info-design-diagonal} gives an upper bound to \eqref{info-design-atomic:main} for $\nmesgs \geq \nstates$, and hence also for \eqref{eq:info-design-joint}. 

\begin{remark}
\label{rem:public-private-connection}
It is interesting to compare the formulations in \eqref{info-design-atomic:main} and \eqref{eq:info-design-pub} for $\nmesgs$-atomic private signaling policies and public signaling policies with $\nmesgs$ messages respectively. 
\ifthenelse{\equal{\isarxiv}{1}}
{While the next result implies that every public signaling policy with $\nmesgs$ messages 
can be equivalently realized by an $\nmesgs$-atomic private signaling policy, the converse is not true in general.


\begin{proposition}
\label{prop:pub-to-private}
Given a $\pfrac \in [0,1]$, for every public signaling policy $\pubsignal$ with $\nmesgs$ messages, there exists an $\nmesgs$-atomic direct private signaling policy with the same cost.
\end{proposition}
\begin{proof}
Let $x=(x^{(1)},\ldots,x^{(\nmesgs)})$ and $y$ be a set of link flows induced by $\pubsignal$, e.g., as given by Proposition~\ref{prop:revelation}.
Consider the private signaling policy $\signal=\pubsignal$ supported on atoms at $(x^{(1)},\ldots,x^{(\nmesgs)})$. Therefore, \eqref{info-design-atomic:nash} is satisfied by the same $y$ as in \eqref{eq:BNE-public-nonreceiving}. Summing \eqref{eq:BNE-public-receiving} over $k \in [\nmesgs]$ gives \eqref{info-design-atomic:obedience}. That is, $(\signal, x,y)$ is feasible for \eqref{info-design-atomic:main}. The equality of expected costs is trivial given that the induced link flows $(x,y)$ and the signaling policies $\signal=\pubsignal$ are the same.
\end{proof}

\begin{remark}
\label{rem:feasible-non-empty}
Proposition~\ref{prop:pub-to-private} implies that, for every $\pfrac \in [0,1]$, there exists a feasible 1-atomic private signaling policy corresponding to $\signal^{\text{pub, no}}$ in \eqref{eq:pub-signal-eg} with $\nmesgs=1$. Therefore, \eqref{eq:info-design-joint} is feasible for every $\pfrac \in [0,1]$. Considering $\nstates$ duplicates of the same atom as for $\nmesgs=1$ case implies that \eqref{eq:info-design-diagonal} is feasible for all $\pfrac \in [0,1]$. Feasibility of \eqref{info-design-atomic:main} can be established along similar lines.
\end{remark}
}{
\revisionchange{Every public signaling policy with $\nmesgs$ messages 
can be equivalently realized by an $\nmesgs$-atomic private signaling policy, but the converse is not true in general. Formally, given a $\pfrac \in [0,1]$, for every public signaling policy $\pubsignal$ with $\nmesgs$ messages, there exists an $\nmesgs$-atomic direct private signaling policy with the same cost.
In particular, there exists a feasible 1-atomic private signal corresponding to $\signal^{\text{pub, no}}$ in \eqref{eq:pub-signal-eg} with $\nmesgs=1$, and hence \eqref{eq:info-design-joint} is feasible for every $\pfrac \in [0,1]$. Considering $\nstates$ duplicates of the same atom as for $\nmesgs=1$ case implies that \eqref{eq:info-design-diagonal} is feasible for all $\pfrac \in [0,1]$. Feasibility of \eqref{info-design-atomic:main} can be established along similar lines.
} 
}
\end{remark}

\begin{remark}
If $\pfrac=1$, i.e., $y=\zerobf$, then Proposition~\ref{prop:gpm-sdp-exact} can be strengthened by identifying tractable sufficient conditions on the coefficients of link latency functions such that \eqref{info-design-gmp-hetero:main} admits a diagonal atomic optimal solution, using the approach of the proof of Proposition~\ref{thm:diagonal-optimal}. Specifically, one can rewrite \eqref{info-design-gmp-hetero:main} only in terms of scalar $x_1$ and then note that it is sufficient to ensure that $A_{\state}^{(i,j)} zz^T$ in \eqref{info-design-gmp-hetero:obedience} is concave in $x_1$ for all $i, j \in [\nlinks=2]$ and $\state \in \allstates$. The latter can be written in terms of non-negativity of polynomials corresponding to the second derivative of $A_{\state}^{(i,j)} zz^T$. These in turn can be equivalently written as semidefinite constraints, e.g., using \cite[Theorems 9 and 10]{Nesterov:00}.
\end{remark}
The next result establishes the equivalence between \eqref{eq:info-design-diagonal} and \eqref{eq:info-design-joint} in a special case, and also establishes that \eqref{eq:info-design-diagonal} is equivalent to the following semidefinite program: 
\begin{subequations}
\label{gmp-to-sdp-hetero:main}
\begin{align}
\min_{M \succeq 0} & \, \, \hat{J}(M):=  C \cdot M \label{gmp-to-sdp-hetero:cost} \\ 
\text{s.t.} \quad &  A^{(i,j)} \cdot M  \geq 0, \quad i, j \in [\nlinks] \label{gmp-to-sdp-hetero:obedience}
\\
&  B^{(i,j)} \cdot M  \geq 0, \quad i, j \in [\nlinks] \label{gmp-to-sdp-hetero:nash}
\\
& M(1,1) = 1  \label{gmp-to-sdp-hetero:unit-mass}\\
& M(i,j) \geq 0, \qquad i, j \in [(\nstates+1)\nlinks +1]  \label{gmp-to-sdp-hetero:non-negative}\\
& S_x^{(k)} \cdot M = 0, \quad S_y \cdot M = 0, \qquad k \in [\nmesgs] \label{gmp-to-sdp-hetero:simplex}\\
& T_x^{(i,k)} \cdot M=0, \quad T_y^{(i)} \cdot M = 0 \qquad i \in [\nlinks], \, k \in [\nmesgs]
\label{gmp-to-sdp-hetero:second-moment} 
\end{align}
\end{subequations}
where the expressions for symmetric matrices $C$, $A^{(i,j)}$, $B^{(i,j)}$, $S_x^{(k)}$, $S_y$, $T_x^{(i,k)}$ and $T_y^{(i)}$ for the special case $D=1$ are provided in \ifthenelse{\equal{\isarxiv}{1}}
{Appendix~\ref{sec:matrix-expressions}}{\revisionchange{the extended version~\cite{Zhu.Savla:infodesign-arxiv20}}}.

\begin{proposition}
\label{thm:diagonal-optimal}
If $\nlinks=2$, then \eqref{eq:info-design-diagonal}, \eqref{eq:info-design-joint} and \eqref{gmp-to-sdp-hetero:main} are all equivalent to each other for \eqref{eq:latency-affine} with $D=1$, i.e., for affine link latency functions. 
\end{proposition}

\begin{remark}
\label{rem:diagonal-atomic}
\begin{enumerate}
\item[(i)] For $\nlinks=2$ and $D=1$, Proposition~\ref{thm:diagonal-optimal} implies that an optimal signaling policy can be realized with $\nstates$ atoms, which is much less than the bound ${3 \choose 2} \nstates=3 \nstates$ given by Theorem~\ref{prop:natoms-upper-bound}. 
\item[(ii)] In spite of its apparent limited scope, Proposition~\ref{thm:diagonal-optimal} is the first to provide a complete characterization of solution to the information design problem even in the most basic setting. 
Indeed, Proposition~\ref{thm:diagonal-optimal} and its proof approach might appear to be generalization of an observation in \cite{Tavafoghi.Teneketzis:20}, which was made for $\pfrac=1$, and for limited affine link latency functions. Not only do we remove these restrictions, but more importantly, our proof implicitly highlights that the obedience constraint needs more careful treatment than suggested in \cite{Tavafoghi.Teneketzis:20}. 
Finally, one can follow the proof of Proposition~\ref{thm:diagonal-optimal} to show that \eqref{eq:info-design-joint} admits a diagonal atomic optimal signaling policy also for $\nlinks=2$ and $D=2$ if $\alpha_{2,\state,1}=\alpha_{2,\state,2} \geq 0$ for all $\state \in \allstates$.
\item[(iii)] It is informative to contrast the different approaches of Proposition~\ref{prop:gpm-sdp-exact} and Proposition~\ref{thm:diagonal-optimal} for establishing tightness of the natural semidefinite relaxation of the corresponding variants of the information design problem. Proposition~\ref{prop:natoms-upper-bound} simply relies on the ability to rewrite the problem in terms of univariate probability measures with compact support, whereas Proposition~\ref{thm:diagonal-optimal} relies on the tightness of the GPM obtained by relaxation because it has optimal probability measures supported on single atoms.
\end{enumerate}
\end{remark}

\subsection{Monotonicity of Optimal Cost Value under Diagonal Atomic Private Signaling Policies}
\label{sec:monotone}
Let $J^{\diag}(x,y)$ denote the cost function in \eqref{eq:info-design-diagonal:cost}, and let $J^{\diag,*}(\pfrac)$ denote the optimal value for a given $\pfrac$. 

\begin{theorem}
\label{thm:cost-monotonicity}
$J^{\diag,*}(\pfrac)$ is continuous and monotonically non-increasing with respect to $\pfrac \in [0,1]$.
\end{theorem}

\begin{remark}
\begin{enumerate}
\item[(i)] Note that Theorem~\ref{thm:cost-monotonicity} does not require the link latency functions to be polynomial.
\item[(ii)] In light of Proposition~\ref{thm:diagonal-optimal}, Theorem~\ref{thm:cost-monotonicity} implies that, if $\nlinks=2$ and if the link latency functions are affine, then the optimal cost value under all, i.e., not necessarily (diagonal) atomic, private signaling policies is continuous and monotonically non-increasing in $\pfrac \in [0,1]$. However, this is not necessarily the case with public signaling policies, as we illustrate in Section~\ref{sec:simulations}.
\item[(iii)] The proof of Theorem~\ref{thm:cost-monotonicity}, in Appendix~\ref{sec:proof-non-increasing}, implies that for a (not necessarily optimal) atomic diagonal signaling policy $\signal^{\diag}$ for some $\pfrac_1 \in [0,1]$, one can construct a simple $\pfrac$-dependent atomic diagonal signaling policy with the same social cost as $\signal^{\diag}$ for all $\pfrac \in [\pfrac_1,1]$. In other words, one can construct a simple feedback (using $\pfrac$) atomic diagonal signaling policy around a \emph{nominal} $\signal^{\diag}$ under which the social cost does not increase due to higher than nominal fraction of participating agents for which $\signal^{\diag}$ is designed. 
This is to be contrasted with existing results according to which the cost of participating agents may increase with their increasing fraction under a fixed (open-loop) signaling policy, e.g., see \cite{Mahmassani.Jayakrishnan:91,Wu.Amin.ea:18}.
\end{enumerate}
\end{remark}

\revisionchange{
\section{Extension to Non-Parallel Networks and Computational Complexity}
\label{sec:extension}

The computational framework for information design extends to non-parallel networks with a single origin-destination pair. A route in this case potentially consists of multiple links. Accordingly, the obedience condition in \eqref{obedience-hetero-v2:obed} becomes $\sum_{\state} \int_x \, \sum_{e \in i} \, \ell_{\state,e}(\tilde{x}_e + \tilde{y}_e) \, x_i \, \signal_{\state}(x) \de x \, \prior(\state)  \leq \sum_{\state} \int_x \, \sum_{e \in j} \, \ell_{\state,e}(\tilde{x}_e + \tilde{y}_e) \, x_i \, \signal_{\state}(x) \de x\, \prior(\state)$ for all $i, j \in [\nlinks]$, where $e \in i$ denotes all the links $e$ constituting route $i$, and $\tilde{x}_e$ and $\tilde{y}_e$ denote the flow on link $e$ induced by the participating and non-participating agents respectively. $\tilde{x}$ and $\tilde{y}$ are linear in $x$ and $y$ respectively. 
Therefore, \eqref{obedience-hetero-v2:obed} is polynomial for polynomial link latency functions also for non-parallel networks. 
The same observation holds true for \eqref{obedience-hetero-v2:nash}, as well as for the cost function in \eqref{info-design:main}. This extends to the indirect policy setup as well. 
Indeed, such generalizations are used in the simulations for the wheatstone network in Section~\ref{sec:simulations}. In particular, for networks with $\nlinks$ routes and link latency functions of degree $D$, the polynomials involved in cost functions and constraints are of degree $D+1$ and have $\nlinks$ variables. Therefore, the lower bound on the number of atoms in Theorem~\ref{prop:natoms-upper-bound} holds as is for non-parallel networks (note that $\nlinks$ is the number of routes and not the number of links).

The optimal signaling policy is to be computed offline, and the messages or recommendations to be generated in real-time as the state changes are obtained through sampling from a given policy. Nevertheless, it is important to examine the complexity of computing optimal policy. 
The worst-case complexity of semidefinite program, under standard solution methods, scales no worse than the square root of problem size, e.g., see \cite{Vandenberghe.Boyd:96}. The sum of the sizes of all the variables in \eqref{info-design-gmp-hetero:main} is $\nlinks(\nmesgs+1)+\nstates \nmesgs$. The maximum degree of the monomials for writing the $r$-th semidefinite relaxation when the link latency functions are of degree $D$ is $\tilde{D}+r-1$, where $\tilde{D}=(D+1)/2$ if $D$ is odd and $=D/2+1$ if $D$ is even. Therefore, the worst case complexity for solving the $r$-th semidefinite relaxation of \eqref{info-design-atomic:main} grows no worse than $\left(\nlinks(\nmesgs+1) + \nstates \nmesgs\right)^{\frac{\tilde{D}+r-1}{2}}$. Semidefinite relaxations of small order have been found in practice to give reasonable lower bound for polynomial optimization problems~\cite{Lasserre:01}, and the practical performance of semidefinite solvers has been found to be much better than indicated by the worse-case complexity~\cite{Vandenberghe.Boyd:96}. 

}

\ifthenelse{\equal{\isarxiv}{1}}
{
\section{Simulations}
\label{sec:simulations}
We compare the minimum cost achievable under private signals, public signals, and full information over two parallel links for affine and BPR latency \ifthenelse{\equal{\numofcolumns}{1}}
{functions (Section~\ref{sec:parallel}) and over a Wheatstone network for affine and quadratic latency functions (Section~\ref{sec:wheatstone}).}{functions.} We also provide a practical scaling of runtime with network size for parallel networks (Section~\ref{sec:runtime}). 
The simulations were performed using a combination of \texttt{GloptiPoly} and the \texttt{MultiStart} function (with \texttt{fmincon} solver) in MATLAB \revisionchange{on a high performance computing facility}.\footnote{The simulation code is available at \url{https://github.com/YixianZhu2016/Information-Design-Simulations}.} In particular, the upper bound computed by \texttt{MultiStart} allows to certify optimality of the lower bound obtained from \texttt{GloptiPoly}, especially when the solution from \texttt{GloptiPoly} does not come with an explicit certificate of optimality. In all the instances in Sections~\ref{sec:parallel} and \ref{sec:wheatstone}, it was found sufficient to have  
100 starting points for \texttt{MultiStart} and relaxation order of 3 for \texttt{GloptiPoly}. We chose $2$ starting points in Section~\ref{sec:runtime}. 
The no information signal corresponds to $\pfrac=0$, when all the costs are expectedly equal. 

\subsection{Parallel Network}
\label{sec:parallel}

For both the scenarios in this case, the total demand is set to be $5$. 
\begin{figure}[htb!]
\begin{center}
\begin{minipage}[c]{.475\textwidth}
\begin{center}
\includegraphics[width=0.85\textwidth]{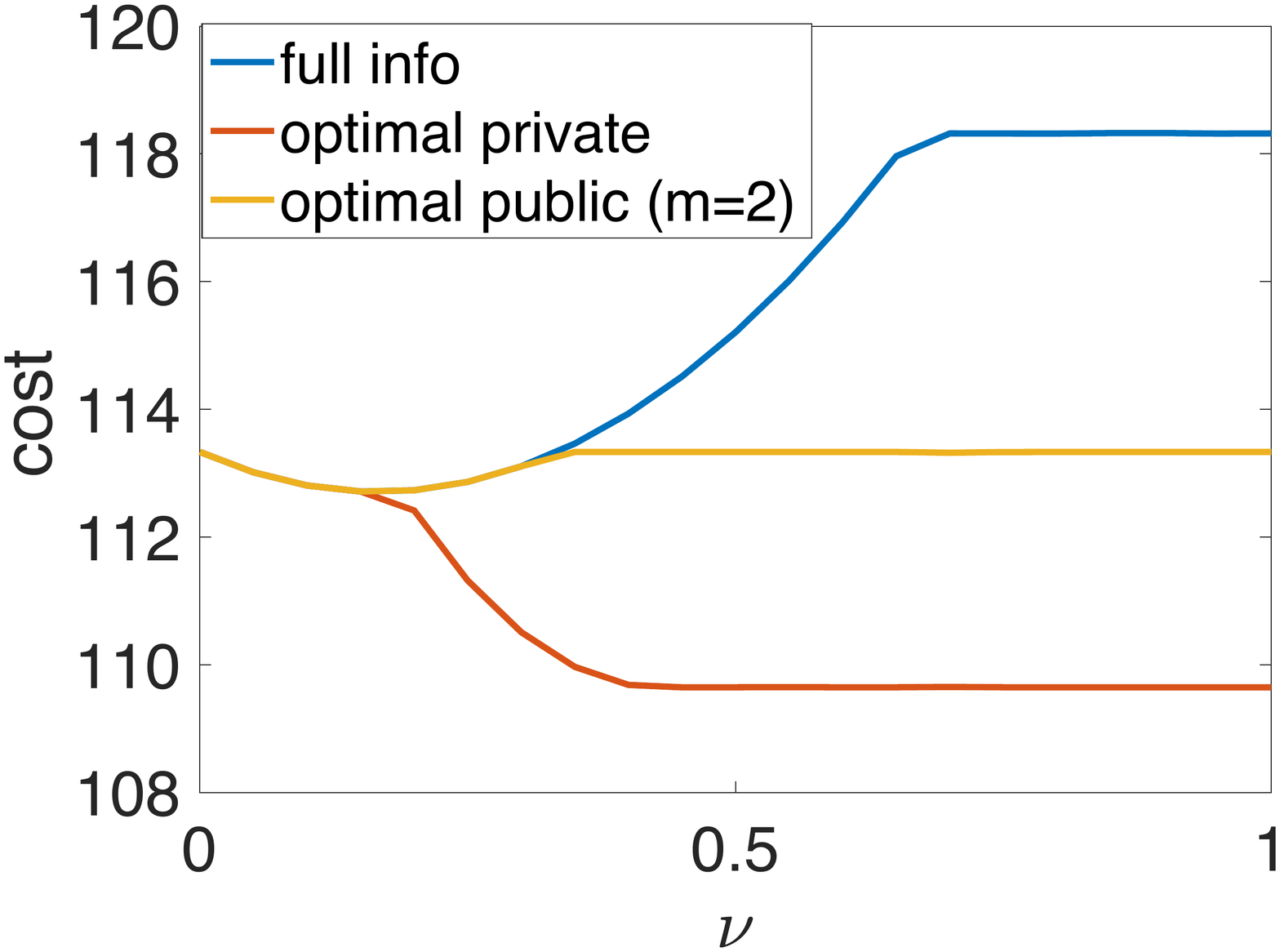} 
\\
(a)
\end{center}
\end{minipage}
\begin{minipage}[c]{.4\textwidth}
\begin{center}
\includegraphics[width=1.0\textwidth]{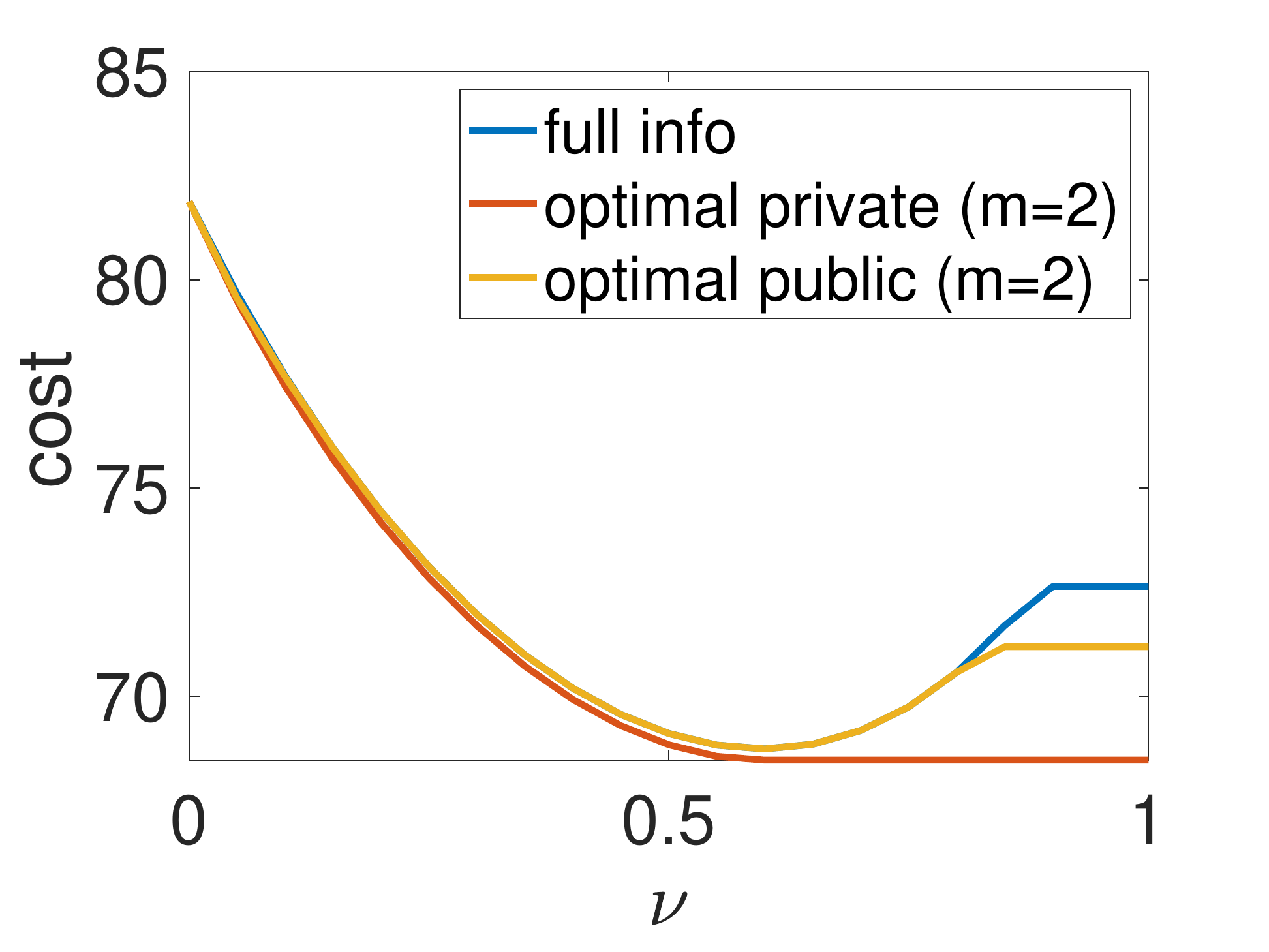}
\\
(b)
\end{center}
\end{minipage}
\end{center}
\caption{\sf Comparison of minimum cost achievable under private signals, public signals and full information over two parallel links, under different $\pfrac$ for (a) affine latency functions and (b) BPR latency functions.}
\label{fig:cost-comparison}
\end{figure}
%

\subsubsection{Affine Latency Functions}
Figure~\ref{fig:cost-comparison}(a) provides comparison between social costs for the following parameters: 
 $\prior(\state_1)=0.6=1-\prior(\state_2)$,
 \ifthenelse{\equal{\numofcolumns}{1}}
{ 
\begin{equation*}
\alpha_0 = 
\kbordermatrix{&i=1&i=2\\
\state_1&5&25\\
\state_2&20&15}, \qquad 
\alpha_1 = 
\kbordermatrix{&i=1&i=2\\
\state_1&4&2\\
\state_2&1&2}
\end{equation*}
}{
$\alpha_0 = 
\kbordermatrix{&i=1&i=2\\
\state_1&5&25\\
\state_2&20&15}$,  
$\alpha_1 = 
\kbordermatrix{&i=1&i=2\\
\state_1&4&2\\
\state_2&1&2}$. 
}
The minimum social cost, i.e., the social cost when the planner can mandate which route every (receiving as well as non-receiving) agent takes for every realization of $\state$,\footnote{This is also referred to as the \emph{first-best} strategy.} for these parameters is $83.33$.
Following Proposition~\ref{thm:diagonal-optimal}, optimal private signal is computed using \eqref{eq:info-design-diagonal}. The approximation to optimal social cost under public signals using \eqref{eq:info-design-pub} was found to be identical for $\nmesgs=2, 3, 4$.
Optimal public signals underlying Figure~\ref{fig:cost-comparison}(a) for $\pfrac = 0.25, 0.5, 0.75, 1$ are, respectively:
\begin{equation*}
\begin{split}
x&=\kbordermatrix{&k=1&k=2\\
i=1&1.25&0\\
i=2&0&1.25}, \,
y=
\begin{bmatrix}
3.23 \\ 0.52
\end{bmatrix}
, \,
\pubsignal=\kbordermatrix{&k=1&k=2\\
\state_1&1&0\\
\state_2&0&1}
\\
x&=\kbordermatrix{&k=1&k=2\\
i=1&2.06&2.06\\
i=2&0.44&0.44}, \,
y=
\begin{bmatrix}
2.11 \\ 0.39
\end{bmatrix}
, \,
\pubsignal=\kbordermatrix{&k=1&k=2\\
\state_1&1&0\\
\state_2&1&0}
\\
x&=\kbordermatrix{&k=1&k=2\\
i=1&3.75&0\\
i=2&0&3.75}, \,
y=
\begin{bmatrix}
0.42 \\ 0.83
\end{bmatrix}
, \,
\pubsignal=\kbordermatrix{&k=1&k=2\\
\state_1&1&0\\
\state_2&1&0}
\\
x&=\kbordermatrix{&k=1&k=2\\
i=1&4.17&0.2\\
i=2&0.83&4.8}, \,
y=
\begin{bmatrix}
0 \\ 0
\end{bmatrix}
, \, 
\pubsignal=\kbordermatrix{&k=1&k=2\\
\state_1&1&0\\
\state_2&1&0}
\end{split}
\end{equation*}

and optimal private signals for the same $\pfrac$ are:
\begin{equation*}
\begin{split}
x&=\kbordermatrix{&k=1&k=2\\
i=1&0.32&0\\
i=2&0.93&1.25}, \,
y=
\begin{bmatrix}
3.75 \\ 0
\end{bmatrix}
, \,
\signal=\kbordermatrix{&k=1&k=2\\
\state_1&1&0\\
\state_2&0&1}
\\
x&=\kbordermatrix{&k=1&k=2\\
i=1&1.58 & 0.37\\
i=2&0.92 &2.13}, \,
y=
\begin{bmatrix}
2.5 \\ 0
\end{bmatrix}
, \, 
\signal=\kbordermatrix{&k=1&k=2\\
\state_1&1&0\\
\state_2&0&1}
\\
x&=\kbordermatrix{&k=1&k=2\\
i=1&2.83&1.62\\
i=2&0.92&2.13}, \,
y=
\begin{bmatrix}
1.25 \\ 0
\end{bmatrix}
, \, 
\signal=\kbordermatrix{&k=1&k=2\\
\state_1&1&0\\
\state_2&0&1}
\\
x&=\kbordermatrix{&k=1&k=2\\
i=1&4.08&2.87\\
i=2&0.92&2.13}, \,
y=
\begin{bmatrix}
0 \\ 0
\end{bmatrix}
, \, 
\signal=\kbordermatrix{&k=1&k=2\\
\state_1&1&0\\
\state_2&0&1}
\end{split}
\end{equation*}

While the cost in Figure~\ref{fig:cost-comparison}(a) shows non-monotonic behavior with respect to $\pfrac$ in the full information case as well as under optimal public signal, the optimal cost is monotonically non-decreasing under private signals. Expectedly, the optimal cost under public signal is no greater than the cost under full information, and the optimal cost under private signal is no greater than under public signal. Interestingly, in this case, full information is an optimal public signal for small values of $\pfrac$, and gives the same cost as an optimal private signal for even smaller values of $\pfrac$.

\subsubsection{BPR Latency Functions}
\label{sec:sims-bpr}
Figure~\ref{fig:cost-comparison}(b) provides comparison between social costs for the following parameters:  $\prior(\state_1)=0.6=1-\prior(\state_2)$, 
\begin{equation*}
\alpha_0 = 
\kbordermatrix{&i=1&i=2\\
\state_1&5&25\\
\state_2&20&15}, \qquad 
\alpha_4 = 
\kbordermatrix{&i=1&i=2\\
\state_1&0.047&0.025\\
\state_2&0.037&0.058}
\end{equation*}
and $\alpha_1 = \alpha_2 = \alpha_3 = \zerobf$. 
These parameters correspond to free flow travel times and capacities being equal to $\alpha_0$ and $\kbordermatrix{&i=1&i=2\\
\state_1&2&3.5\\
\state_2&3&2.5}$ respectively. The minimum social cost for these parameters is $52.78$.

The approximation to optimal social cost under private signals using \eqref{info-design-atomic:main} was found to be identical for $\nmesgs=2,3,4$, suggesting that $\nmesgs=2$ atoms are possibly sufficient to realize optimal private signal in this case. This is much less than the upper bound of $2 {6 \choose 5}=12$ atoms given by Theorem~\ref{prop:natoms-upper-bound}.
Similarly, the approximation to optimal social cost under public signals using \eqref{eq:info-design-pub} was found to be identical for $\nmesgs=2,3,4$. 
Optimal public signals underlying Figure~\ref{fig:cost-comparison}(b) for $\pfrac=0.25, 0.5, 0.75, 1$ are, respectively:
\begin{equation*}
\begin{split}
x&=\kbordermatrix{&k=1&k=2\\
i=1&1.25&0\\
i=2&0&1.25}, \,
y=
\begin{bmatrix}
3.75 \\ 0
\end{bmatrix}
, \, 
\pubsignal=\kbordermatrix{&k=1&k=2\\
\state_1&1&0\\
\state_2&0&1}
\\
x&=\kbordermatrix{&k=1&k=2\\
i=1&2.5&0\\
i=2&0&2.5}, \,
y=
\begin{bmatrix}
2.5 \\ 0
\end{bmatrix}
, \,
\pubsignal=\kbordermatrix{&k=1&k=2\\
\state_1&1&0\\
\state_2&0&1}
\\
x&=\kbordermatrix{&k=1&k=2\\
i=1&3.75&0\\
i=2&0&3.75}, \,
y=
\begin{bmatrix}
1.25 \\ 0
\end{bmatrix}
, \, 
\pubsignal=\kbordermatrix{&k=1&k=2\\
\state_1&1&0\\
\state_2&0&1}
\\
x&=\kbordermatrix{&k=1&k=2\\
i=1&5.0&2.08\\
i=2&0.0&2.92}, \,
y=
\begin{bmatrix}
0 \\ 0
\end{bmatrix}
, \,
\pubsignal=\kbordermatrix{&k=1&k=2\\
\state_1&0.87&0\\
\state_2&0.13&1}
\end{split}
\end{equation*}

and optimal private signals for the same $\pfrac$ are:
\begin{equation*}
\begin{split}
x&=\kbordermatrix{&k=1&k=2\\
i=1&0.99&0\\
i=2&0.26&1.25}, \,
y=
\begin{bmatrix}
3.75 \\ 0
\end{bmatrix}
, \, 
\signal=\kbordermatrix{&k=1&k=2\\
\state_1&1&0\\
\state_2&0&1}
\\
x&=\kbordermatrix{&k=1&k=2\\
i=1&2.24&0.0\\
i=2&0.26&2.5}, \,
y=
\begin{bmatrix}
2.5 \\ 0
\end{bmatrix}
, \, 
\signal=\kbordermatrix{&k=1&k=2\\
\state_1&1&0\\
\state_2&0&1}
\\
x&=\kbordermatrix{&k=1&k=2\\
i=1&3.49&0.76\\
i=2&0.26&2.99}, \,
y=
\begin{bmatrix}
1.25 \\ 0
\end{bmatrix}
, \, 
\signal=\kbordermatrix{&k=1&k=2\\
\state_1&1&0\\
\state_2&0&1}
\\
x&=\kbordermatrix{&k=1&k=2\\
i=1&4.74&2.01\\
i=2&0.26&2.99}, \,
y=
\begin{bmatrix}
0 \\ 0
\end{bmatrix}
, \, 
\signal=\kbordermatrix{&k=1&k=2\\
\state_1&1&0\\
\state_2&0&1}
\end{split}
\end{equation*}
The social cost profile in Figure~\ref{fig:cost-comparison}(b) shows similar qualitative dependence on $\pfrac$ as in Figure~\ref{fig:cost-comparison}(a). Since diagonal atomic private signals are observed to be optimal (based on the sample values reported above), monotonicity of the corresponding cost is consistent with Theorem~\ref{thm:cost-monotonicity}.


\subsection{Wheatstone Network}
\label{sec:wheatstone}

\ifthenelse{\equal{\numofcolumns}{1}}
{
\begin{figure}[htb!]
\begin{center}
\begin{minipage}[c]{.3\textwidth}
\begin{center}
\begin{tikzpicture}[node distance=3cm]

[fill=red!50!white,draw=black] (-4.8,1.6) circle(0.5);

\path (4,0) node(a) [circle,draw,fill=blue!50!white] {o}
	(6,1) node(b) [circle,draw,fill=blue!50!white] {$  $}
	(6,-1) node(c) [circle,draw,fill=blue!50!white] {$  $}
	(8,0) node(d) [circle,draw,fill=blue!50!white] {d};

\draw[very thick,->,line width=0.6mm] (a) -- (b);
\draw[very thick,->,line width=0.6mm] (a) -- (c); 
\draw[very thick,->,line width=0.6mm] (b) -- (c);
\draw[very thick,->,line width=0.6mm] (b) -- (d);
\draw[very thick,->,line width=0.6mm] (c) -- (d);    

\draw (4.75,0.75) node {$i=1$};
\draw (7.25,0.75) node {$i=2$};
\draw (4.75,-0.75) node {$i=3$};
\draw (7.25,-0.75) node {$i=4$};
\draw (6.5,0) node {$i=5$};
	
\end{tikzpicture} 
\end{center}
\end{minipage}
\begin{minipage}[c]{.3\textwidth}
\begin{center}
\includegraphics[width=1\linewidth]{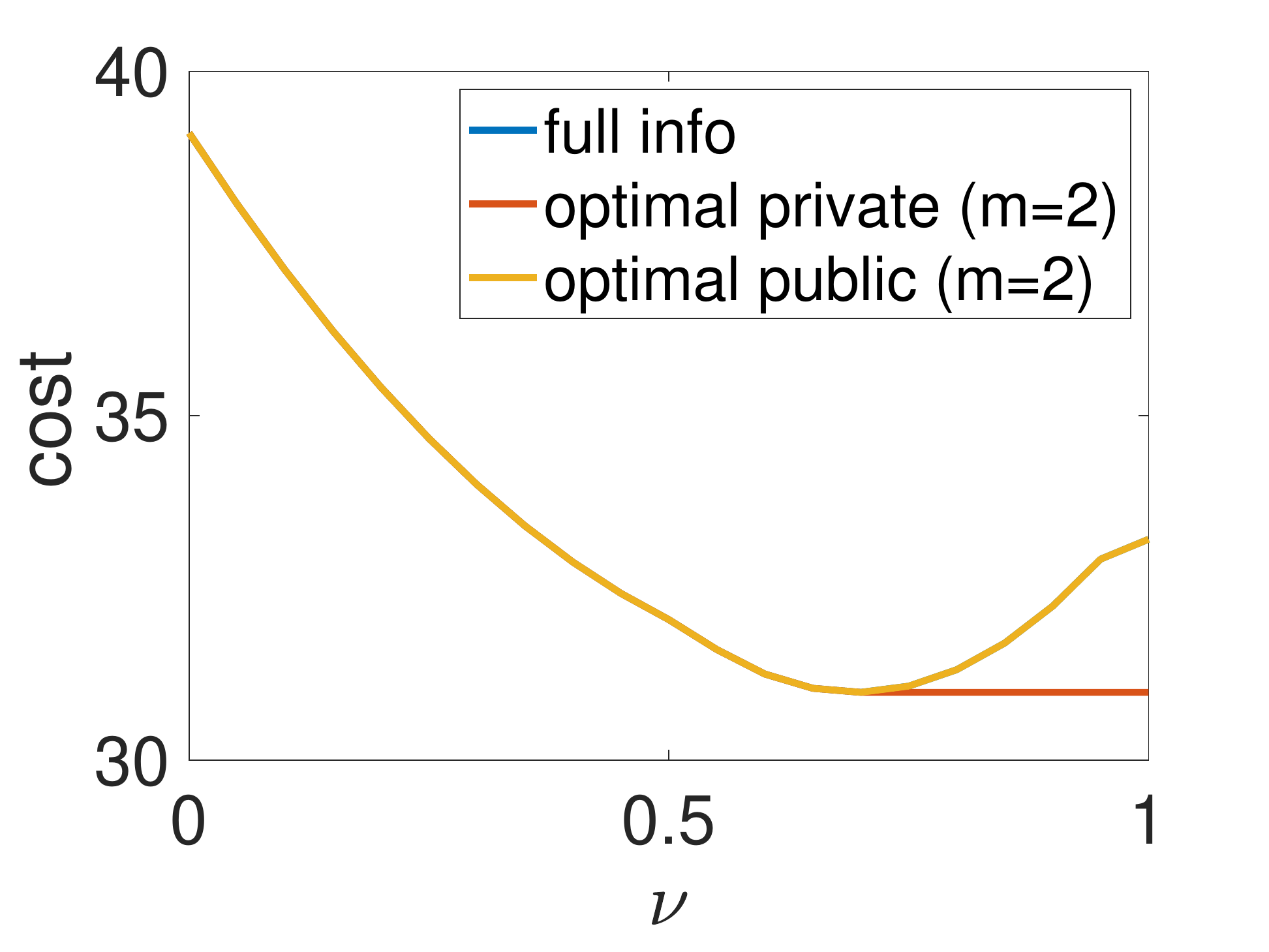} 
\end{center}
\end{minipage}
\begin{minipage}[c]{.3\textwidth}
\begin{center}
\includegraphics[width=0.9\linewidth]{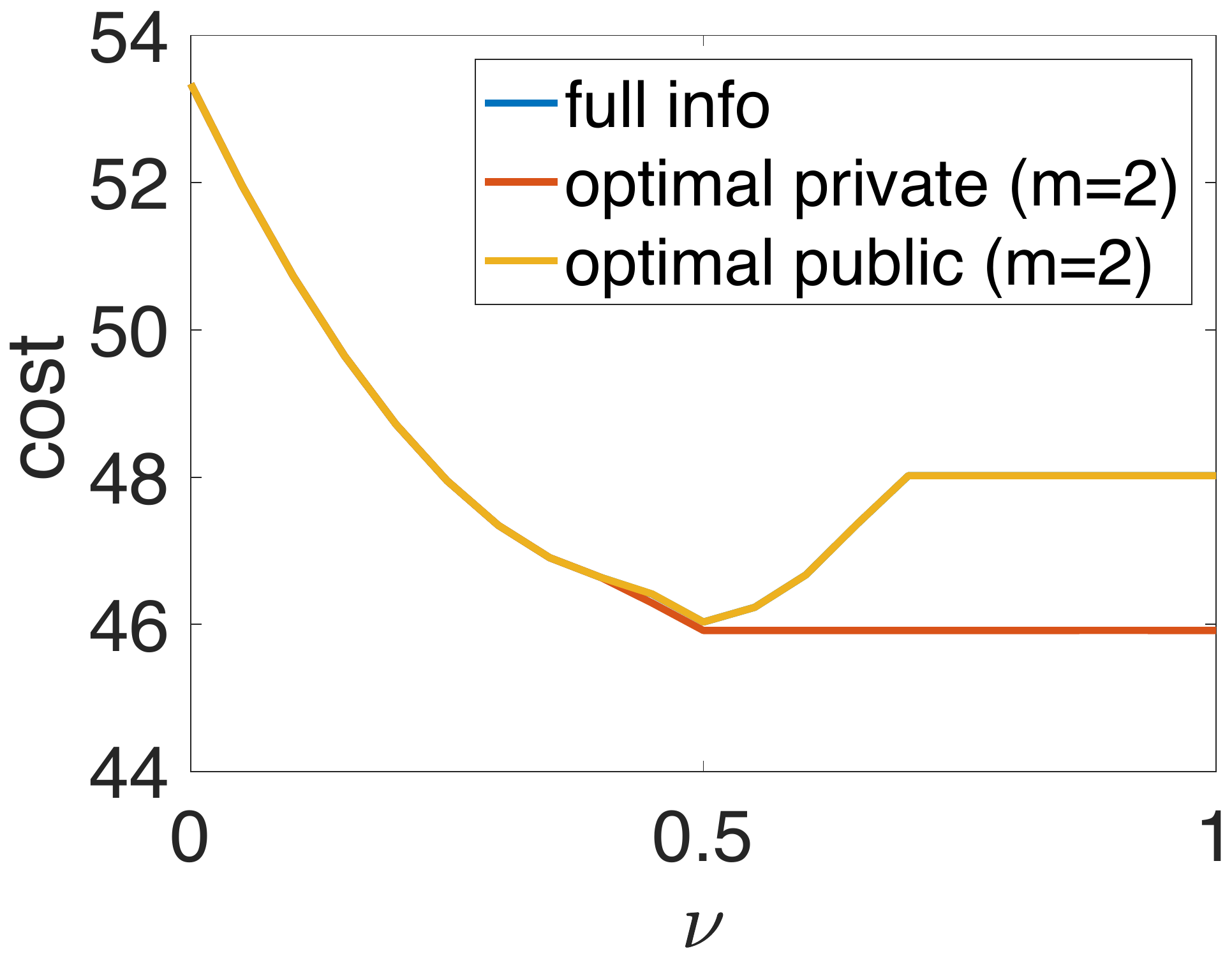} 
\end{center}
\end{minipage}
\begin{minipage}[c]{.3\textwidth}
\begin{center}
(a)
\end{center}
\end{minipage}
\begin{minipage}[c]{.3\textwidth}
\begin{center}
(b)
\end{center}
\end{minipage}
\begin{minipage}[c]{.3\textwidth}
\begin{center}
(c)
\end{center}
\end{minipage}
\end{center}
\caption{\sf (a) Wheatstone network; (b) comparison of social costs under affine link latency functions; (c) comparison of social costs under quadratic link latency functions.}
\label{fig:wheatstone}
\end{figure}
}{
\begin{figure}[htb!]
\begin{center}
\begin{minipage}[c]{.2\textwidth}
\begin{center}
\begin{tikzpicture}[node distance=3cm]

[fill=red!50!white,draw=black] (-4.8,1.6) circle(0.5);

\path (4,0) node(a) [circle,draw,fill=blue!50!white] {o}
	(6,1) node(b) [circle,draw,fill=blue!50!white] {$  $}
	(6,-1) node(c) [circle,draw,fill=blue!50!white] {$  $}
	(8,0) node(d) [circle,draw,fill=blue!50!white] {d};

\draw[very thick,->,line width=0.6mm] (a) -- (b);
\draw[very thick,->,line width=0.6mm] (a) -- (c); 
\draw[very thick,->,line width=0.6mm] (b) -- (c);
\draw[very thick,->,line width=0.6mm] (b) -- (d);
\draw[very thick,->,line width=0.6mm] (c) -- (d);    

\draw (4.75,0.75) node {$i=1$};
\draw (7.25,0.75) node {$i=2$};
\draw (4.75,-0.75) node {$i=3$};
\draw (7.25,-0.75) node {$i=4$};
\draw (6.5,0) node {$i=5$};
	
\end{tikzpicture} 
\\
(a)
\end{center}
\end{minipage}
\begin{minipage}[c]{.275\textwidth}
\begin{center}
\includegraphics[width=0.9\linewidth]{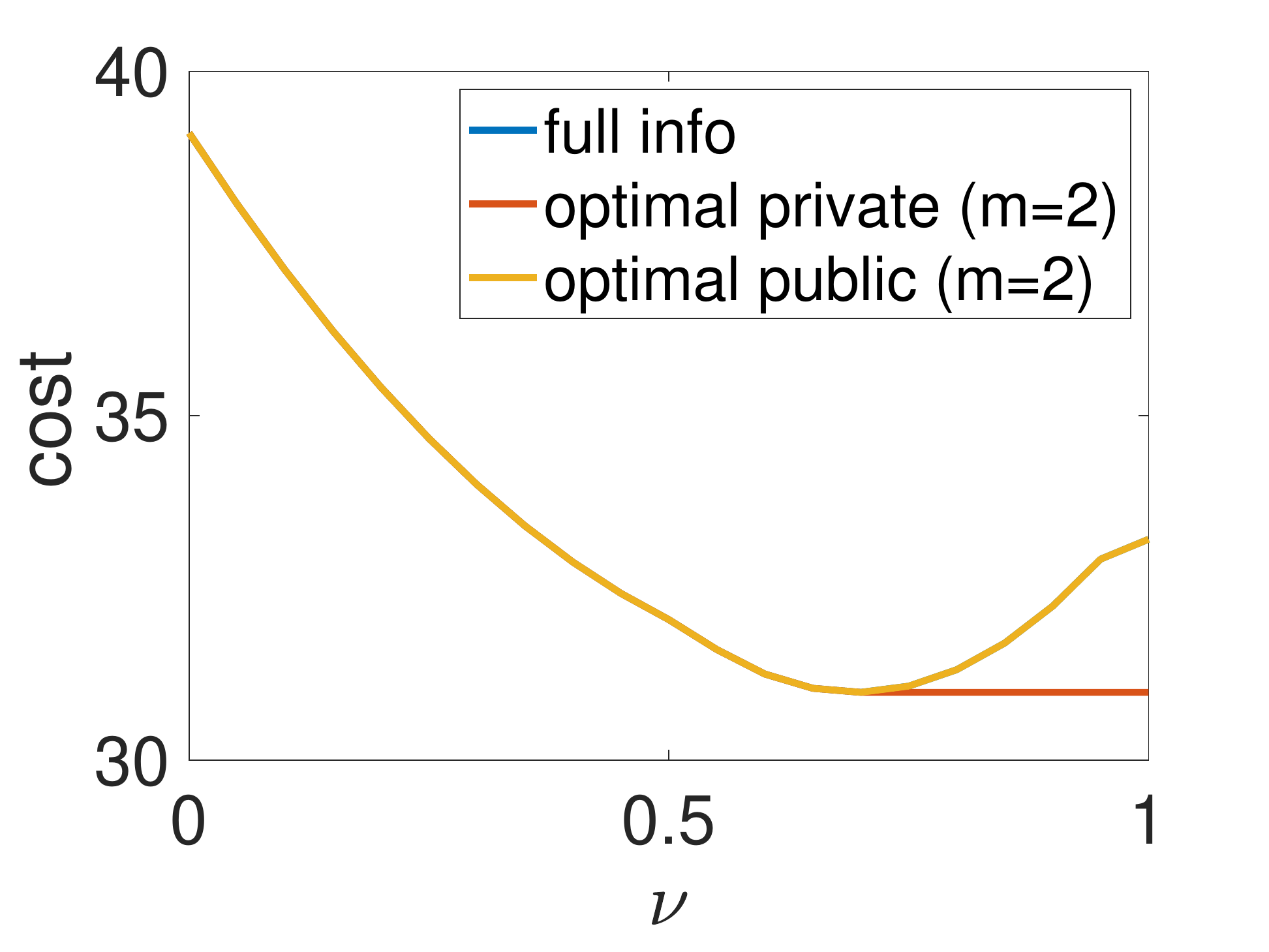} \\
(b)
\end{center}
\end{minipage}
\begin{minipage}[c]{.275\textwidth}
\begin{center}
\includegraphics[width=0.9\linewidth]{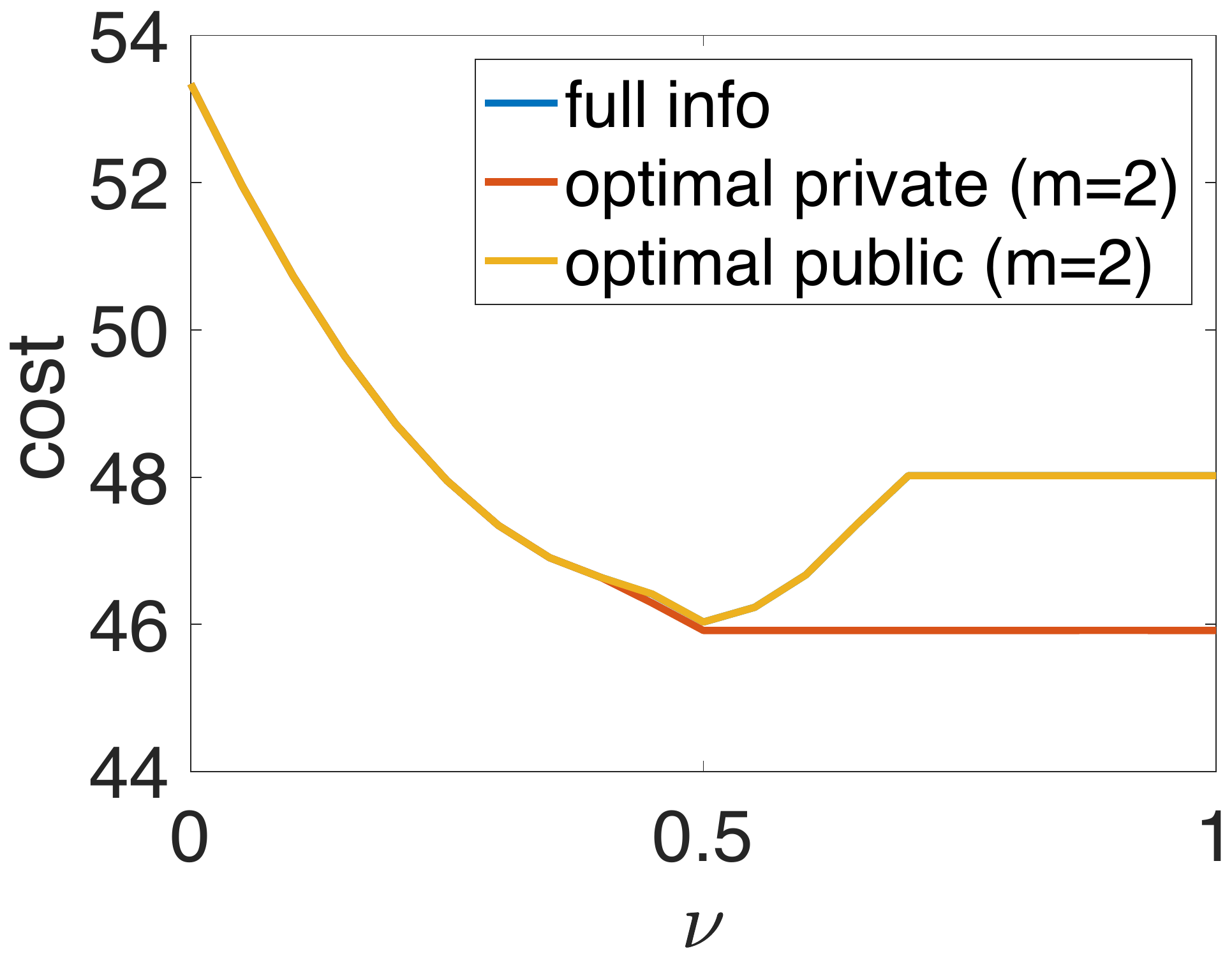} \\
(c)
\end{center}
\end{minipage}
\end{center}
\caption{\sf (a) Wheatstone network; (b) comparison of social costs under affine link latency functions; (c) comparison of social costs under quadratic link latency functions.}
\label{fig:wheatstone}
\end{figure}
}

\subsubsection{Affine Latency Functions}
Consider the Wheatstone network shown in Figure~\ref{fig:wheatstone}(a), where a demand of 2.5 needs to be routed from $o$ to $d$. Figure~\ref{fig:wheatstone}(b) shows comparison between the costs for the following simulation parameters: $\prior(\state_1)=0.5=1-\prior(\state_2)$, 
\begin{equation*}
\begin{split}
\alpha_0 & =\kbordermatrix{&i=1&i=2 & i=3 & i=4 & i=5\\
\state_1&1 & 15 & 24 & 1 & 2\\
\state_2&1 & 0.5 & 4 & 1 & 20},
\\
\alpha_1 & =\kbordermatrix{&i=1&i=2 & i=3 & i=4 & i=5\\
\state_1&3 & 1 & 1.5 & 0.5 & 5\\
\state_2&3 & 0.5 & 1.5 & 0.5 & 5}
\end{split}
\end{equation*}

Consider paths 1, 2 and 3 consisting of $i=\{1,2\}$, $i=\{3,4\}$ and $i=\{1,5,4\}$ respectively. 
The minimum social cost for these simulation parameters is $19.67$. 
The optimal social cost under public and private signals for $\nmesgs=2$ atoms are plotted in Figure~\ref{fig:wheatstone}(b). Optimal public signals for $\pfrac=0.25, 0.5, 0.75, 1$ are, respectively:
\begin{equation*}
\begin{split}
x&=\kbordermatrix{&k=1&k=2\\
\text{path 1}&0&0.625\\
\text{path 2}&0.625&0\\
\text{path 3}&0&0
}, \,
y=
\begin{bmatrix}
1.53 \\ 0.34 \\ 0
\end{bmatrix}
, \,
\pubsignal=\kbordermatrix{&k=1&k=2\\
\state_1&1&0\\
\state_2&0&1}
\\
x&=\kbordermatrix{&k=1&k=2\\
\text{path 1}&0&1.25\\
\text{path 2}&1.25&0\\
\text{path 3}&0&0
}, \,
y=
\begin{bmatrix}
1.23 \\ 0.02 \\ 0
\end{bmatrix}
, \,
\pubsignal=\kbordermatrix{&k=1&k=2\\
\state_1&1&0\\
\state_2&0&1}
\\
x&=\kbordermatrix{&k=1&k=2\\
\text{path 1}&0&1.875\\
\text{path 2}&1.875&0\\
\text{path 3}&0&0
}, \,
y=
\begin{bmatrix}
0.625 \\ 0 \\ 0
\end{bmatrix}
, \,
\pubsignal=\kbordermatrix{&k=1&k=2\\
\state_1&1&0\\
\state_2&0&1}
\\
x&=\kbordermatrix{&k=1&k=2\\
\text{path 1}&0.08&2.5\\
\text{path 2}&2.42&0\\
\text{path 3}&0&0
}, \,
y=\begin{bmatrix}
0 \\ 0 \\ 0
\end{bmatrix}
, \, 
\pubsignal=
\kbordermatrix{&k=1&k=2\\
\state_1&1&0\\
\state_2&0&1}
\end{split}
\end{equation*}

and a set of optimal private signals for the same $\pfrac$ are:
\begin{equation*}
\begin{split}
x&=\kbordermatrix{&k=1&k=2\\
\text{path 1}&0.02&0.61\\
\text{path 2}&0.61&0.02\\
\text{path 3}&0&0
}, \,
y=
\begin{bmatrix}
1.53 \\ 0.34 \\ 0
\end{bmatrix}
, \, 
\signal=\kbordermatrix{&k=1&k=2\\
\state_1&1&0\\
\state_2&0&1}
\\
x&=\kbordermatrix{&k=1&k=2\\
\text{path 1}&0&1.25\\
\text{path 2}&1.25&0\\
\text{path 3}&0&0
}, \,
y=
\begin{bmatrix}
1.25 \\ 0 \\ 0
\end{bmatrix}
, \, 
\signal=\kbordermatrix{&k=1&k=2\\
\state_1&1&0\\
\state_2&0&1}
\\
x&=\kbordermatrix{&k=1&k=2\\
\text{path 1}&0.14&1.87\\
\text{path 2}&1.73&0\\
\text{path 3}&0&0
}, \,
y=
\begin{bmatrix}
0.63 \\ 0 \\ 0
\end{bmatrix}
, \, 
\signal=\kbordermatrix{&k=1&k=2\\
\state_1&1&0\\
\state_2&0&1}
\\
x&=\kbordermatrix{&k=1&k=2\\
\text{path 1}&0.76&2.5\\
\text{path 2}&1.74&0\\
\text{path 3}&0&0
}, \,
y=
\begin{bmatrix}
0 \\ 0 \\ 0
\end{bmatrix}
, \,
\signal=\kbordermatrix{&k=1&k=2\\
\state_1&1&0\\
\state_2&0&1}
\end{split}
\end{equation*}

The social cost profile in Figure~\ref{fig:wheatstone}(b) shows similar qualitative dependence on $\pfrac$ as in Figure~\ref{fig:cost-comparison}(a), with the exception that the full information signal is an optimal public signal for all $\pfrac \in [0,1]$ in this case.

%

\subsubsection{Quadratic Latency Functions}
Consider the same Wheatstone network setup as before, except for quadratic link latency functions with the following coefficients:
\begin{equation*}
\begin{split}
\alpha_0 & =\kbordermatrix{&i=1&i=2 & i=3 & i=4 & i=5\\
\state_1&1 & 15 & 24 & 1 & 2\\
\state_2&1 & 0.5 & 4 & 1 & 20}, \quad 
\alpha_1 =\kbordermatrix{&i=1&i=2 & i=3 & i=4 & i=5\\
\state_1&2 & 3 & 5 & 4 & 2\\
\state_2&4 & 3 & 1 & 4 & 3},
\\
\alpha_2 & =\kbordermatrix{&i=1&i=2 & i=3 & i=4 & i=5\\
\state_1&0.4314 & 0.1818 & 0.1455 & 0.8693 & 0.5499\\
\state_2&0.9106 & 0.2638 & 0.1361 & 0.5797 & 0.1450}
\end{split}
\end{equation*}

The minimum social cost in this case is found to be $29.40$. The optimal social cost under public and private signals for $\nmesgs=2$ atoms are plotted in Figure~\ref{fig:wheatstone}(c). Optimal public signals for $\pfrac=0.25, 0.5, 0.75, 1$ are, respectively:
\begin{equation*}
\begin{split}
x&=\kbordermatrix{&k=1&k=2\\
\text{path 1}&0&0\\
\text{path 2}&0&0.625\\
\text{path 3}&0.625&0
}, \,
y=
\begin{bmatrix}
1.521 \\ 0.354 \\ 0
\end{bmatrix}
, \,
\pubsignal=\kbordermatrix{&k=1&k=2\\
\state_1&1&0\\
\state_2&0&1}
\\
x&=\kbordermatrix{&k=1&k=2\\
\text{path 1}&0&0.017\\
\text{path 2}&0&1.233\\
\text{path 3}&1.25&0
}, \,
y=
\begin{bmatrix}
1.25 \\ 0 \\ 0
\end{bmatrix}
, \,
\pubsignal=\kbordermatrix{&k=1&k=2\\
\state_1&1&0\\
\state_2&0&1}
\\
x&=\kbordermatrix{&k=1&k=2\\
\text{path 1}&0.160&0.642\\
\text{path 2}&0&1.233\\
\text{path 3}&1.715&0
}, \,
y=
\begin{bmatrix}
0.625 \\ 0 \\ 0
\end{bmatrix}
, \,
\pubsignal=\kbordermatrix{&k=1&k=2\\
\state_1&1&0\\
\state_2&0&1}
\\
x&=\kbordermatrix{&k=1&k=2\\
\text{path 1}&0.785&1.267\\
\text{path 2}&0&1.233\\
\text{path 3}&1.715&0
}, \,
y=\begin{bmatrix}
0 \\ 0 \\ 0
\end{bmatrix}
, \, 
\pubsignal=
\kbordermatrix{&k=1&k=2\\
\state_1&1&0\\
\state_2&0&1}
\end{split}
\end{equation*}

and a set of optimal private signals for the same $\pfrac$ are:
\begin{equation*}
\begin{split}
x&=\kbordermatrix{&k=1&k=2\\
\text{path 1}&0&0\\
\text{path 2}&0&0.625\\
\text{path 3}&0.625&0
}, \,
y=
\begin{bmatrix}
1.521 \\ 0.354 \\ 0
\end{bmatrix}
, \, 
\signal=\kbordermatrix{&k=1&k=2\\
\state_1&1&0\\
\state_2&0&1}
\\
x&=\kbordermatrix{&k=1&k=2\\
\text{path 1}&0.025&0.04\\
\text{path 2}&0.108&1.21\\
\text{path 3}&1.117&0
}, \,
y=
\begin{bmatrix}
1.25 \\ 0 \\ 0
\end{bmatrix}
, \, 
\signal=\kbordermatrix{&k=1&k=2\\
\state_1&1&0\\
\state_2&0&1}
\\
x&=\kbordermatrix{&k=1&k=2\\
\text{path 1}&0.653&0.664\\
\text{path 2}&0.104&1.211\\
\text{path 3}&1.118&0
}, \,
y=
\begin{bmatrix}
0.625 \\ 0 \\ 0
\end{bmatrix}
, \, 
\signal=\kbordermatrix{&k=1&k=2\\
\state_1&1&0\\
\state_2&0&1}
\\
x&=\kbordermatrix{&k=1&k=2\\
\text{path 1}&1.277&1.290\\
\text{path 2}&0.108&1.210\\
\text{path 3}&1.115&0
}, \,
y=
\begin{bmatrix}
0 \\ 0 \\ 0
\end{bmatrix}
, \,
\signal=\kbordermatrix{&k=1&k=2\\
\state_1&1&0\\
\state_2&0&1}
\end{split}
\end{equation*}

The social cost profile in Figure~\ref{fig:wheatstone}(c) shows similar qualitative dependence on $\pfrac$ as in Figure~\ref{fig:cost-comparison}(a), with the exception that the full information signal is an optimal public signal for all $\pfrac \in [0,1]$ in this case.

%

\subsection{Scaling of Runtime with Network Size}
\label{sec:runtime}

\ifthenelse{\equal{\numofcolumns}{1}}
{
\begin{figure}[htb!]
\begin{center}
\begin{minipage}[c]{.4\textwidth}
\begin{center}
\includegraphics[width=0.9\linewidth]{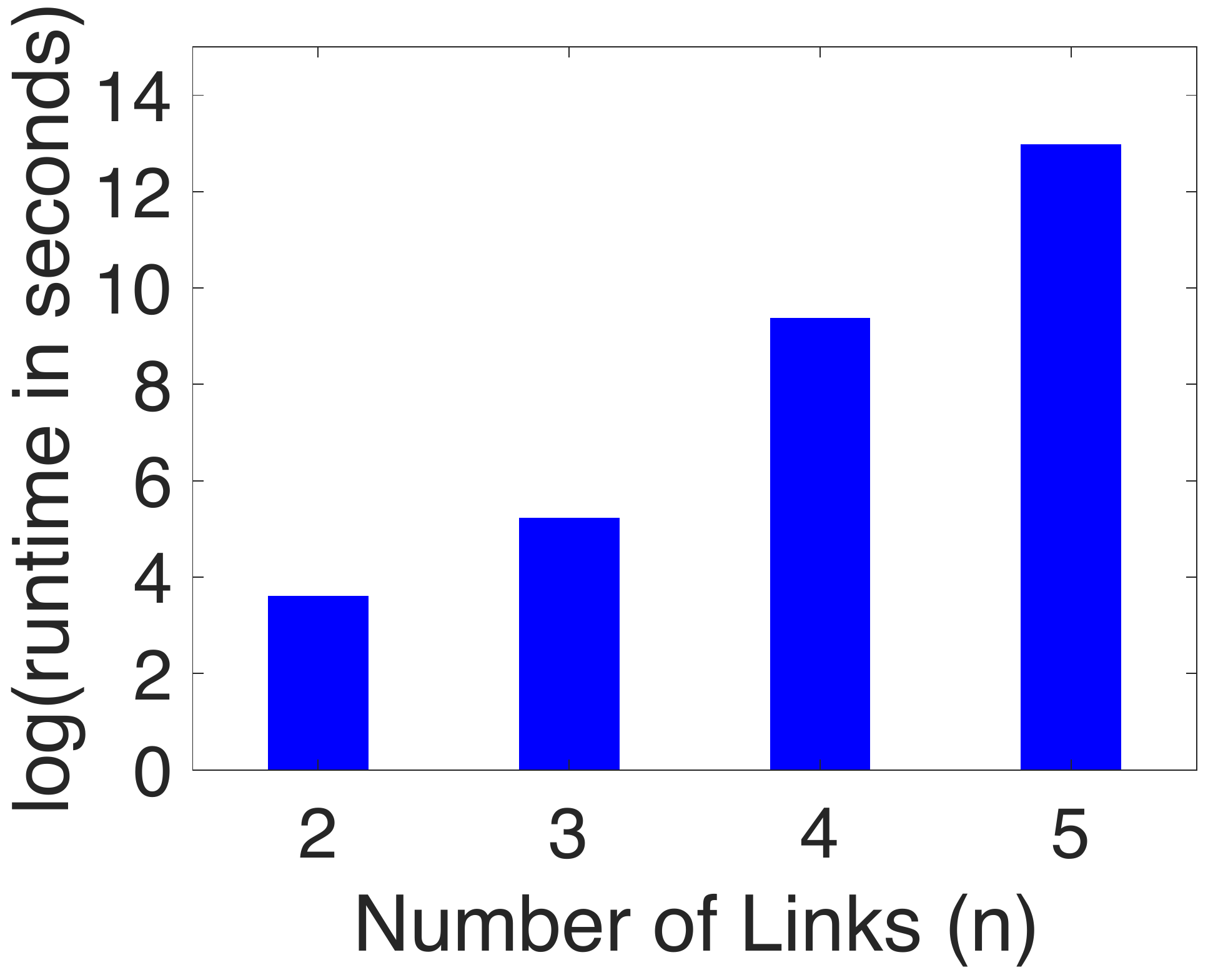}
\end{center}
\end{minipage}
\end{center}
\caption{\sf Log linear plot of run time versus $\nlinks$ for parallel networks.}
\label{fig:computation_complexity}
\end{figure}
}{
\begin{figure}[htb!]
\begin{center}
\begin{minipage}[c]{.2\textwidth}
\begin{center}
\includegraphics[width=0.9\linewidth]{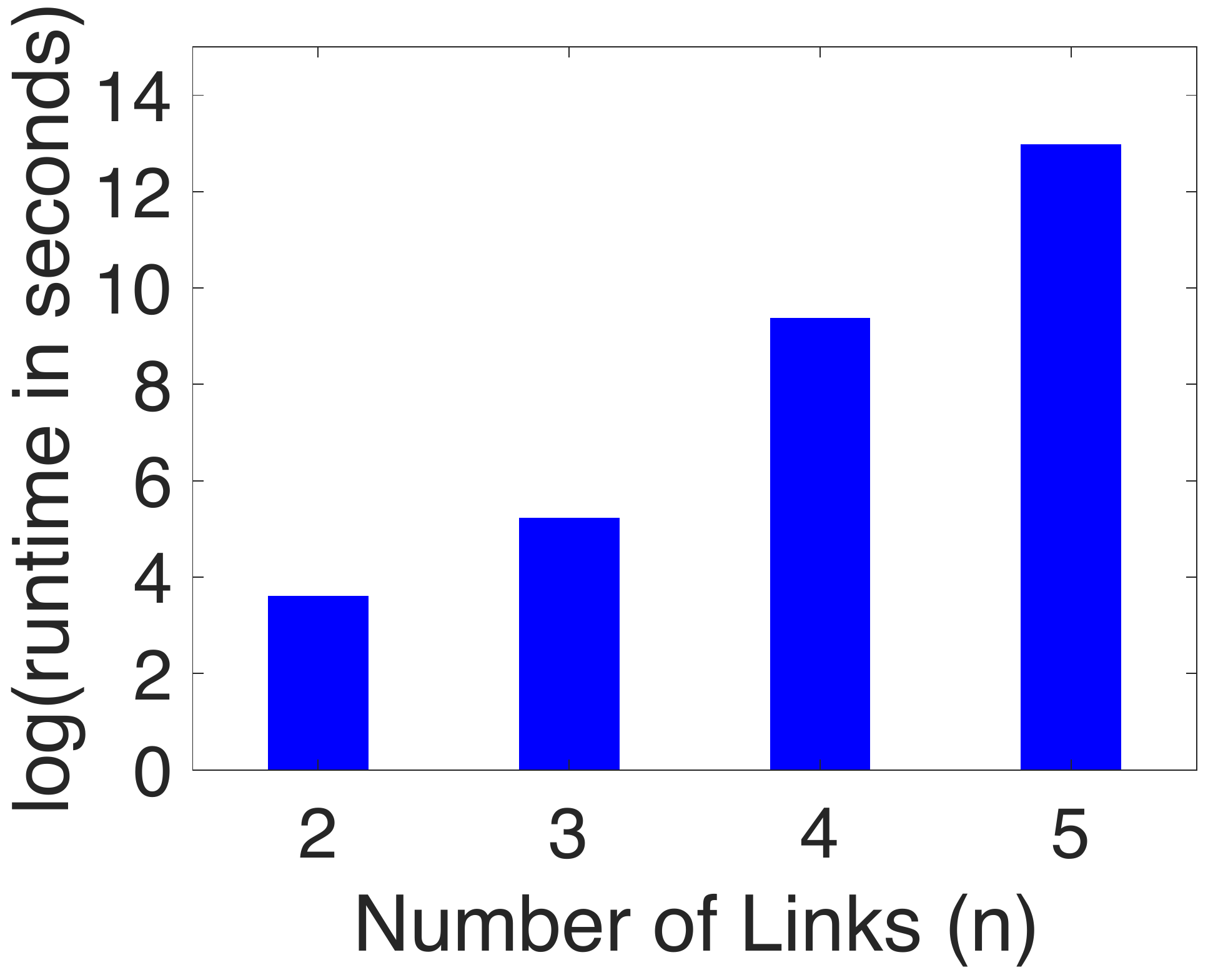} 
\end{center}
\end{minipage}
\end{center}
\caption{\sf Log linear plot of run time versus $\nlinks$ for parallel networks.}
\label{fig:computation_complexity}
\end{figure}
}
We revisit the parallel network setup and report runtime versus number of links. The link latency functions are affine with coefficients for $\nlinks$-link network to be the first $\nlinks$ columns of 
\begin{equation*}
\begin{split}
\alpha_0 & =\kbordermatrix{&i=1&i=2 & i=3 & i=4 & i=5\\
\state_1&5 & 25 & 4 & 24 & 17\\
\state_2&20 & 15 & 24 & 12 & 19},
\\
\alpha_1 & =\kbordermatrix{&i=1&i=2 & i=3 & i=4 & i=5\\
\state_1&4 & 2 & 1 & 2 & 4\\
\state_2&1 & 2 & 3 & 5 & 2}
\end{split}
\end{equation*}
The total demand is set to be $2.5 \nlinks$, and the prior is $\prior(\state_1)=0.6=1-\prior(\state_2)$ in all the instances. The log-linear plot in Figure~\ref{fig:computation_complexity} suggests that runtime grows exponentially with $\nlinks$. This apparent inconsistency with the discussion in Section~\ref{sec:extension} is to be understood in the context of complex resource management strategies embedded in the high performance computing facility used for these simulations. Furthermore, implicit in the analysis of  Section~\ref{sec:extension} is a large $\nlinks$ assumption, and a fixed absolute accuracy level which is independent of $\nlinks$. Relaxing these assumptions in the context of practical solvers and hardware limitations is outside the scope of this paper, and will be pursued in future work.  

%
%
%

}
{
\input{simulations-v3}
}

\section{Conclusion and Future Work}
\label{sec:conclusions}
Existing works on information design for non-atomic routing games provide useful insights, whose generalization however is not readily apparent. 
Relatedly, a computational approach to operationalize optimal information design for general settings does not exist to the best of our knowledge. By making connection to semidefinite programming (SDP), this paper not only fills this gap, but also allows to leverage computational tools developed by the SDP community. The latter is particularly relevant for extending the approach to non-atomic games beyond routing.   

There are several directions for future work. The bound in Theorem~\ref{prop:natoms-upper-bound} may be computationally prohibitive for large networks. Proposition~\ref{thm:diagonal-optimal}, related discussion in Remark~\ref{rem:diagonal-atomic}, and Section~\ref{sec:simulations} on the other hand suggest the possibility of exploring problem structure to tighten the bound. A counterpart to Theorem~\ref{prop:natoms-upper-bound} for public signaling policies is open. A relatively unexplored direction is sub-optimality bounds for simple classes of signaling policies such as diagonal atomic. 
%
%
 Finally, it would be interesting to utilize the approach in this paper to quantify the reduction in \emph{price of anarchy} under information design. This will complement, e.g., preliminary analysis in \cite{Vasserman.Feldman.ea:15}.

\vspace{-0.025in}
\section*{Acknowledgment}
The high performance computing support provided by USC's Center for Advanced Research Computing for running simulations is gratefully acknowledged.

\vspace{-0.025in}
\bibliographystyle{ieeetr}
\bibliography{../bib/ksmain,../bib/savla}

\appendix

\subsection{Matrix Expressions}
\label{sec:matrix-expressions}
In the \ifthenelse{\equal{\numofcolumns}{1}}
{matrices below,}{matrices,} the lower triangular entries, generically represented as *, are equal to their upper triangular counterparts, and $e_{i}$ is the standard $i$-th basis vector in $\RR^n$, i.e., its $i$-th entry is one and all the other entries are zero. 

\ifthenelse{\equal{\isarxiv}{1}}
{
Expressions for matrices in \eqref{info-design-gmp-hetero:main} when $D=1$ are as follows:
\begin{equation*}
\label{eq:C-matrix-expr}
\begin{split}
C_{\state}(y) & = \prior(\state)
\begin{bmatrix}
y^T \diag(\alpha_{1,\state})y + y^T \alpha_{0,\state}
& \frac{{\alpha_{0,\state}}^T}{2} + y^T \diag(\alpha_{1,\state}) \\
*  & \diag(\alpha_{1,\state})
\end{bmatrix}, \quad 
\alpha_{d,\state} =[\alpha_{d,\state,1}, \ldots, \alpha_{d,\state,\nlinks}]^T, \quad d=0,1
\end{split}
\end{equation*}

Expressions for matrices in \eqref{gmp-to-sdp-hetero:main} when $D=1$ are as follows:

}{
}

\ifthenelse{\equal{\numofcolumns}{1}}
{
\begin{minipage}[c]{.79\textwidth}
\begin{equation*}
\begin{split}
C & = 
\kbordermatrix{ & &x^{\state_1}& & x^{\state_{\nstates}} & y\\
& 0 & \frac{\prior(\state_1)}{2} {\alpha_{0,\state_1}}^T & \ldots & \frac{\prior(\state_{\nstates})}{2} {\alpha_{0,\state_{\nstates}}}^T & \frac{{\overline{\alpha}_0}^T}{2} \\
x^{\state_1} & * & \prior(\state_1) \diag(\alpha_{1,\state_1})  & \ldots & \zerobf & \prior(\state_1)\diag(\alpha_{1,\state_1}) \\
& \vdots & \vdots & \ldots & \vdots & \vdots  \\
x^{\state_{\nstates}} & * &  * &  \ldots & \prior(\state_{\nstates})\diag(\alpha_{1,\state_{\nstates}}) & \prior(\state_{\nstates})\diag(\alpha_{1,\state_{\nstates}})  \\
y & * & * &  \ldots & * & \diag(\overline{\alpha}_1) 
}, 
\end{split}
\end{equation*}
\end{minipage}
\begin{minipage}[c]{.2\textwidth}
\begin{equation*}
\begin{split}
 \overline{\alpha}_d & :=\sum_{\state} \prior(\state) \, \alpha_{d,\state} 
  \\
 \alpha_{d,\state} & := [\alpha_{d,\state,1}, \ldots, \alpha_{d,\state,\nlinks}]^T
  \end{split}
  \end{equation*}
  \end{minipage}
  }{
  \begin{figure*}[!htb]
  \begin{center}
  \scalebox{.8}{
  \begin{minipage}[c]{.75\textwidth}
\begin{equation*}
\begin{split}
C & = 
\kbordermatrix{ & &x^{\state_1}& & x^{\state_{\nstates}} & y\\
& 0 & \frac{\prior(\state_1)}{2} {\alpha_{0,\state_1}}^T & \ldots & \frac{\prior(\state_{\nstates})}{2} {\alpha_{0,\state_{\nstates}}}^T & \frac{{\overline{\alpha}_0}^T}{2} \\
x^{\state_1} & * & \prior(\state_1) \diag(\alpha_{1,\state_1})  & \ldots & \zerobf & \prior(\state_1)\diag(\alpha_{1,\state_1}) \\
& \vdots & \vdots & \ldots & \vdots & \vdots  \\
x^{\state_{\nstates}} & * &  * &  \ldots & \prior(\state_{\nstates})\diag(\alpha_{1,\state_{\nstates}}) & \prior(\state_{\nstates})\diag(\alpha_{1,\state_{\nstates}})  \\
y & * & * &  \ldots & * & \diag(\overline{\alpha}_1) 
}, 
\end{split}
\end{equation*}
\end{minipage} \quad
\begin{minipage}[c]{.175\textwidth}
\begin{equation*}
\begin{split}
 \overline{\alpha}_d & :=\sum_{\state} \prior(\state) \, \alpha_{d,\state} 
  \\
 \alpha_{d,\state} & := [\alpha_{d,\state,1}, \ldots, \alpha_{d,\state,\nlinks}]^T
  \end{split}
  \end{equation*}
  \end{minipage}
  }
  \end{center}
  \end{figure*}
  }
  \ifthenelse{\equal{\numofcolumns}{1}}
{
\begin{equation*}
\begin{split}
A^{(i,j)} & = 
\kbordermatrix{ & &x^{\state_1}& & x^{\state_{\nstates}} & y\\
& 0 & \prior(\state_1)\frac{\alpha_{0,\state_1,j}-\alpha_{0,\state_1,i}}{2} e^T_i  & \ldots & \prior(\state_{\nstates}) \frac{\alpha_{0,\state_{\nstates},j}-\alpha_{0,\state_{\nstates},i}}{2} e^T_i & \zerobf  \\
x^{\state_1} & * & \tilde{A}^{(i,j)}_{\state_1} + {\tilde{A}^{(i,j)}_{\state_1}}^T  & \ldots & \zerobf & \tilde{A}^{(i,j)}_{\state_1} \\
 & \vdots &\ldots & \vdots & \vdots  \\
x^{\state_{\nstates}} & * & * & \ldots &\tilde{A}^{(i,j)}_{\state_{\nstates}} + {\tilde{A}^{(i,j)}_{\state_{\nstates}}}^T & \tilde{A}^{(i,j)}_{\state_{\nstates}}  \\
y & * & * & \ldots & * & \zerobf 
} 
\end{split}
\end{equation*}
}{
\begin{figure*}[!htb]
\begin{center}
\scalebox{.8}{
\begin{minipage}[c]{.99\textwidth}
\begin{equation*}
A^{(i,j)} = 
\kbordermatrix{ & &x^{\state_1}& & x^{\state_{\nstates}} & y\\
& 0 & \prior(\state_1)\frac{\alpha_{0,\state_1,j}-\alpha_{0,\state_1,i}}{2} e^T_i  & \ldots & \prior(\state_{\nstates}) \frac{\alpha_{0,\state_{\nstates},j}-\alpha_{0,\state_{\nstates},i}}{2} e^T_i & \zerobf  \\
x^{\state_1} & * & \tilde{A}^{(i,j)}_{\state_1} + {\tilde{A}^{(i,j)}_{\state_1}}^T  & \ldots & \zerobf & \tilde{A}^{(i,j)}_{\state_1} \\
 & \vdots &\ldots & \vdots & \vdots  \\
x^{\state_{\nstates}} & * & * & \ldots &\tilde{A}^{(i,j)}_{\state_{\nstates}} + {\tilde{A}^{(i,j)}_{\state_{\nstates}}}^T & \tilde{A}^{(i,j)}_{\state_{\nstates}}  \\
y & * & * & \ldots & * & \zerobf 
}, \qquad \tilde{A}_{\state}^{(i,j)} = \prior(\state) \left(\frac{\alpha_{1,\state,j}}{2} e_i e_j^T - \frac{\alpha_{1,\state,i}}{2} e_i e_i^T \right) 
\end{equation*}
\end{minipage}
}
\end{center}
\end{figure*}
}
 \ifthenelse{\equal{\numofcolumns}{1}}
{
\begin{equation*}
\begin{split}
\tilde{A}_{\state}^{(i,j)} = \prior(\state) \left(\frac{\alpha_{1,\state,j}}{2} e_i e_j^T - \frac{\alpha_{1,\state,i}}{2} e_i e_i^T \right), \qquad 
B^{(i,j)} = 
\kbordermatrix{ & &x^{\state_1}& & x^{\state_{\nstates}} & y\\
& 0 & \zerobf  & \ldots & \zerobf & \frac{\overline{\alpha}_{0,j} - \overline{\alpha}_{0,i}}{2} e_i^T \\
x^{\state_1} & * & \zerobf & \ldots & \zerobf & \tilde{A}^{(i,j)}_{\state_1} \\
& \vdots & \vdots  &\ldots & \vdots & \vdots  \\
x^{\state_{\nstates}} & * & * &  \ldots & \zerobf & \tilde{A}^{(i,j)}_{\state_{\nstates}}  \\
y & * & * &  \ldots & * & \sum_k \, \tilde{A}^{(i,j)}_{\state_k}
}
\end{split}
\end{equation*}
}{
\begin{equation*}
\begin{split}
B^{(i,j)} = 
\kbordermatrix{ & &x^{\state_1}& & x^{\state_{\nstates}} & y\\
& 0 & \zerobf  & \ldots & \zerobf & \frac{\overline{\alpha}_{0,j} - \overline{\alpha}_{0,i}}{2} e_i^T \\
x^{\state_1} & * & \zerobf & \ldots & \zerobf & \tilde{A}^{(i,j)}_{\state_1} \\
& \vdots & \vdots  &\ldots & \vdots & \vdots  \\
x^{\state_{\nstates}} & * & * &  \ldots & \zerobf & \tilde{A}^{(i,j)}_{\state_{\nstates}}  \\
y & * & * &  \ldots & * & \sum_k \, \tilde{A}^{(i,j)}_{\state_k}
}
\end{split}
\end{equation*}
}
\ifthenelse{\equal{\numofcolumns}{1}}
{
\begin{equation*}
\begin{split}
S_x^{(k)} = 
\kbordermatrix{ & &x^{\state_1}& & x^{\state_k}&  & x^{\state_{\nstates}} & y\\
& - \pfrac & \zerobf & \ldots & \onebf^T/2 & \ldots & \zerobf & \zerobf  \\
x^{\state_1}& * & \zerobf & \ldots & \zerobf & \ldots & \zerobf & \zerobf  \\
& \vdots & \vdots & \ldots & \vdots & \ldots & \vdots & \vdots  \\
x^{\state_k}& * & * & \ldots & \zerobf & \ldots & \zerobf & \zerobf  \\
& \vdots & \vdots & \ldots & \vdots & \ldots & \vdots & \vdots  \\
x^{\state_{\nstates}}& * & * & \ldots & * & \ldots & \zerobf & \zerobf \\
y& * & * & \ldots & * & \ldots & * & \zerobf 
}, \quad 
S_y = 
\kbordermatrix{ & &x^{\state_1}& & x^{\state_{\nstates}} & y\\
& \pfrac - 1 & \zerobf & \ldots & \zerobf & \onebf^T/2  \\
x^{\state_1} & * &  \zerobf & \ldots & \zerobf & \zerobf  \\
& \vdots & \vdots & \ldots & \vdots & \vdots  \\
x^{\state_{\nstates}} & * &  * & \ldots & \zerobf & \zerobf  \\
y & * &  * & \ldots & * & \zerobf 
}
\end{split}
\end{equation*}
\begin{equation*}
\begin{split}
T_x^{(i,k)} = 
\kbordermatrix{ & &x^{\state_1}& & x^{\state_k}&  & x^{\state_{\nstates}} & y\\
& 0 & \zerobf & \ldots & - \frac{\pfrac e_i^T}{2} & \ldots & \zerobf & \zerobf  \\
x^{\state_1} & * & \zerobf & \ldots & \zerobf & \ldots & \zerobf & \zerobf  \\
& \vdots & \vdots & \ldots & \vdots & \ldots & \vdots & \vdots  \\
x^{\state_k}& * & * & \ldots & \frac{\onebf e_i^T + e_i \onebf^T}{2} & \ldots & \zerobf & \zerobf  \\
& \vdots & \vdots & \ldots & \vdots & \ldots & \vdots & \vdots  \\
x^{\state_{\nstates}} & * & * & \ldots & * & \ldots & \zerobf & \zerobf  \\
y & * & * & \ldots & * & \ldots & * & \zerobf 
}, \quad
T_y^{(i)} = 
\kbordermatrix{ & &x^{\state_1}& & x^{\state_{\nstates}} & y\\
& 0 & \zerobf & \ldots & \zerobf & \frac{(1-\pfrac) e_i^T}{2}  \\
x^{\state_1}& * & \zerobf & \ldots & \zerobf & \zerobf  \\
& \vdots & \vdots & \ldots & \vdots & \vdots  \\
x^{\state_{\nstates}}& * & * & \ldots & \zerobf & \zerobf  \\
y& * & * & \ldots & * & \frac{\onebf e_i^T + e_i \onebf^T}{2} 
}
\end{split}
\end{equation*}
}{
\begin{figure*}[!htb]
 \begin{center}
 \scalebox{.85}{
 \begin{minipage}[c]{.99\textwidth}
\begin{equation*}
\begin{split}
S_x^{(k)} & = 
\kbordermatrix{ & &x^{\state_1}& & x^{\state_k}&  & x^{\state_{\nstates}} & y\\
& - \pfrac & \zerobf & \ldots & \onebf^T/2 & \ldots & \zerobf & \zerobf  \\
x^{\state_1}& * & \zerobf & \ldots & \zerobf & \ldots & \zerobf & \zerobf  \\
& \vdots & \vdots & \ldots & \vdots & \ldots & \vdots & \vdots  \\
x^{\state_k}& * & * & \ldots & \zerobf & \ldots & \zerobf & \zerobf  \\
& \vdots & \vdots & \ldots & \vdots & \ldots & \vdots & \vdots  \\
x^{\state_{\nstates}}& * & * & \ldots & * & \ldots & \zerobf & \zerobf \\
y& * & * & \ldots & * & \ldots & * & \zerobf 
}, \quad 
S_y = 
\kbordermatrix{ & &x^{\state_1}& & x^{\state_{\nstates}} & y\\
& \pfrac - 1 & \zerobf & \ldots & \zerobf & \onebf^T/2  \\
x^{\state_1} & * &  \zerobf & \ldots & \zerobf & \zerobf  \\
& \vdots & \vdots & \ldots & \vdots & \vdots  \\
x^{\state_{\nstates}} & * &  * & \ldots & \zerobf & \zerobf  \\
y & * &  * & \ldots & * & \zerobf 
}
\\
T_x^{(i,k)} & = 
\kbordermatrix{ & &x^{\state_1}& & x^{\state_k}&  & x^{\state_{\nstates}} & y\\
& 0 & \zerobf & \ldots & - \frac{\pfrac e_i^T}{2} & \ldots & \zerobf & \zerobf  \\
x^{\state_1} & * & \zerobf & \ldots & \zerobf & \ldots & \zerobf & \zerobf  \\
& \vdots & \vdots & \ldots & \vdots & \ldots & \vdots & \vdots  \\
x^{\state_k}& * & * & \ldots & \frac{\onebf e_i^T + e_i \onebf^T}{2} & \ldots & \zerobf & \zerobf  \\
& \vdots & \vdots & \ldots & \vdots & \ldots & \vdots & \vdots  \\
x^{\state_{\nstates}} & * & * & \ldots & * & \ldots & \zerobf & \zerobf  \\
y & * & * & \ldots & * & \ldots & * & \zerobf 
}, \quad  
T_y^{(i)} = 
\kbordermatrix{ & &x^{\state_1}& & x^{\state_{\nstates}} & y\\
& 0 & \zerobf & \ldots & \zerobf & \frac{(1-\pfrac) e_i^T}{2}  \\
x^{\state_1}& * & \zerobf & \ldots & \zerobf & \zerobf  \\
& \vdots & \vdots & \ldots & \vdots & \vdots  \\
x^{\state_{\nstates}}& * & * & \ldots & \zerobf & \zerobf  \\
y& * & * & \ldots & * & \frac{\onebf e_i^T + e_i \onebf^T}{2} 
}
\end{split}
\end{equation*}
\end{minipage}
}
\end{center}
\end{figure*}
}

\ifthenelse{\equal{\isarxiv}{1}}
{
\subsection{Proof of Proposition~\ref{prop:revelation}}
It is easy to see that \eqref{eq:BNE-flow-solution} is convex since the link latency functions are non-decreasing. Therefore, the KKT conditions for optimality become necessary and sufficient and can be shown to be equivalent to the BNE condition in \eqref{eq:BNE} following standard arguments.

In order to establish the second part of the proposition, consider the following set of aggregate link flows:
\begin{subequations}
\label{eq:agg-flow-set:main}
\begin{align}
\tilde{\mc X} := & \Big\{\tilde{x}^{(k)} \in \simplex(1), \, k \in [\nmesgs]: \, \tilde{x}^{(k)}=x^{(k)}+y,  \, x^{(k)} \in \simplex(\pfrac), \, k \in [\nmesgs], \, y \in \simplex(1-\pfrac) \Big\}
\label{eq:agg-flow-set:def}
\\
= & \Big\{\tilde{x}^{(k)} \in \simplex(1), \, k \in [\nmesgs]: \, \sum_i \min_k \tilde{x}_i^{(k)} \geq 1 - \pfrac
\Big\}
\label{eq:agg-flow-set:reformulation}
\end{align}
\end{subequations}
\eqref{eq:agg-flow-set:reformulation} follows from the following. For every $(\{x^{(k)}: k \in [\nmesgs]\}, y)$, $\sum_i \, \min_k \, (x_i^{(k)} + y_i) \geq \sum_i y_i = 1-\pfrac$. Vice-versa, for every $\tilde{x}$ satisfying the inequality in \eqref{eq:agg-flow-set:reformulation}, let $\tilde{y}_i := \min_k \tilde{x}_i^{(k)}$, $i \in [\nlinks]$. With this, $y_i=\frac{\tilde{y}_i}{\sum_j \tilde{y}_j} (1-\pfrac) \in \simplex(1-\pfrac)$ and $x^{(k)}=\tilde{x}^{(k)}-y \in \simplex(\pfrac)$, $k \in [\nmesgs]$. 

Convexity of $\tilde{\mc X}$ is established as follows. Consider any $\tilde{x}^{(1,k)}$ and $\tilde{x}^{(2,k)}$, $k \in [\nmesgs]$, belonging to $\tilde{\mc X}$. Therefore, for all $\beta \in [0,1]$,
\begin{equation*}
\begin{split}
\sum_i \, \min_k \, \left(\beta \tilde{x}_i^{(1,k)} + (1-\beta) \tilde{x}_i^{(2,k)}  \right) & = \sum_i \, \left(\beta \tilde{x}_i^{(1,k(i))} + (1-\beta) \tilde{x}_i^{(2,k(i))}  \right)
\\
& \geq \beta \sum_i \, \min_k \, \tilde{x}_i^{(1,k)} + (1-\beta) \sum_i \, \min_k \, \tilde{x}_i^{(2,k)}
\\
& = 1-\pfrac
\end{split}
\end{equation*}
where $k(i)$ in the first equality is a $k \in [\nmesgs]$ for which $\beta \tilde{x}_i^{(1,k)} + (1-\beta) \tilde{x}_i^{(2,k)}$ achieves the smallest value. 

Now consider the following:
\begin{equation}
\label{eq:BNE-flow-reformulation}
\begin{split}
\min_{\substack{\{\tilde{x}^{(k)}: \, k \in [\nmesgs]\} \in \tilde{\mc X}}} \quad & \sum_{i, \, \state} \, \, \int_{\bar{x}} \, \, \int_0^{\sum_k \bar{x}_k \, \tilde{x}^{(k)}_{i}} \, \congfunc_{\state,i}(z)  \, \indsignal_{\state}(\bar{x}) \, \prior(\state) \, \de z \, \de \bar{x}
\end{split} 
\end{equation}
where we use \eqref{eq:agg-flow-set:reformulation} for $\tilde{\mc X}$. Recall that $\tilde{\mc X}$ is convex. 
The cost function in \eqref{eq:BNE-flow-reformulation}, say $F(\tilde{x})$, can be shown to be strictly convex for all $\tilde{x} \in \tilde{\mc X}$ as follows. The generic entry of the Hessian of $F$ is given by:
\begin{equation*}
\frac{\partial^2 F}{\partial \tilde{x}_h^{(k_1)}\partial \tilde{x}_j^{(k_2)}} 
= \begin{cases}
\sum_{\state} \, \, \int_{\bar{x}} \, \, \bar{x}_{k_1} \bar{x}_{k_2} \congfunc'_{\state,h}(\sum_k \bar{x}_k \tilde{x}_h^{(k)}) \, \indsignal_{\state}(\bar{x}) \, \prior(\state) \, \de \bar{x}, & h=j
\\
0, & h \neq j
\end{cases}
\end{equation*}
Therefore, $\tilde{x}^T \nabla^2 \tilde{x} = \sum_{i, \, \state} \, \, \int_{\bar{x}} \, \, \left(\sum_k \bar{x}_{k} \tilde{x}_i^{(k)}\right)^2 \congfunc'_{\state,i}(\sum_k \bar{x}_k \tilde{x}_i^{(k)}) \, \indsignal_{\state}(\bar{x}) \, \prior(\state) \, \de \bar{x}>0$ for all $\tilde{x} \in \tilde{\mc X}$, where the inequality holds because the integrand is non-negative, and for every $\bar{x} \in \simplex_{\nmesgs}$, $\sum_k \bar{x}_k \tilde{x}_i^{(k)}>0$ for at least one $i$. 

Therefore, \eqref{eq:BNE-flow-reformulation} is strictly convex. Following the definition of $\tilde{\mc X}$ in \eqref{eq:agg-flow-set:def}, it is also easy to see that \eqref{eq:BNE-flow-reformulation} and \eqref{eq:BNE-flow-solution} give the same solution. Therefore, for every global minimum of \eqref{eq:BNE-flow-solution}, i.e., BNE flow, there exists a unique aggregate link flow in $\tilde{\mc X}$, which is the unique global minimum of \eqref{eq:BNE-flow-reformulation}.


}

\subsection{Proof of Proposition~\ref{prop:gpm-sdp-exact}}
Substituting $x_2=\pfrac-x_1$, \eqref{eq:info-design-joint} can be re-written in terms of probability measures $\tilde{\signal}=\{\tilde{\signal}_{\state}: \, \state \in \allstates\}$ over the single variable $x_1$, with the only constraint that each entry of $\tilde{\signal}$ is supported over $[0,\pfrac]$. Let $\tilde{\mom}:=\{(\tilde{\mom}_{\state}^0,\ldots,\tilde{\mom}_{\state}^{D+1}): \, \state \in \allstates\}$ be the reals corresponding to the first $D+1$ moments of $\tilde{\signal}$. The cost function in \eqref{info-design-gmp-hetero:cost} and the constraint in \eqref{info-design-gmp-hetero:obedience}-\eqref{info-design-gmp-hetero:nash} can be expressed as linear combinations of elements of $\tilde{\mom}$. The additional constraint that the elements of $\tilde{\mom}$ have to correspond to the first $D+1$ moments of probability measures supported on $[0,\pfrac]$ can be written in terms of linear equations and semidefinite matrix constraints, e.g., see \cite[Proposition A.6]{Stein.Parrilo.ea:GEB11}.

\subsection{Proof of Theorem~\ref{prop:natoms-upper-bound}}
We refer to Section~\ref{sec:technical-results} for the definition of a \emph{truncated moment sequence} used in this proof. 

Substituting $x_{\nlinks}=\pfrac-\sum_{i \in [\nlinks-1]}x_i$ and $y_{\nlinks}=\pfrac-\sum_{i \in [\nlinks-1]}y_i$, \eqref{eq:info-design-joint} can be equivalently rewritten in terms of $(x_1,\ldots,x_{\nlinks-1})$ and $(y_1,\ldots,y_{\nlinks-1})$. We use this reduced form of \eqref{eq:info-design-joint} for this proof. 
Let $(\signal^*,y^*)$ be an optimal solution to \eqref{eq:info-design-joint}. 
The polynomials in $x$ appearing in the cost and constraints in \eqref{eq:info-design-joint} have highest degree $D+1$. Consider a $(\tilde{\signal}^*,y^*)$, where, for every $\state \in [\nstates]$, $\tilde{\signal}_{\state}^*$ has the same truncated moment sequence of degree $D+1$ as $\signal_{\state}^*$. Such a $(\tilde{\signal}^*,y^*)$ satisfies the constraints in \eqref{eq:info-design-joint} and gives the same cost value as $(\signal^*,y^*)$, and is therefore also optimal. The theorem then follows from \cite[Theorem 2]{Bayer.Teichmann:06} according to which, a truncated moment sequence in $\nlinks-1$ variables of degree $D+1$ admits a feasible measure if and only if it admits a feasible measure with support consisting of at most ${D+\nlinks \choose D+1}$ atoms. 

\subsection{Proof of Proposition~\ref{thm:diagonal-optimal}}

\subsubsection*{\underline{Equivalence between \eqref{eq:info-design-diagonal} and \eqref{eq:info-design-joint}}} Let $(\signal^*,y^*)$ be an optimal solution to \eqref{eq:info-design-joint}. We show that, for every $y \in \simplex(1-\pfrac)$, there exists an optimal solution to \eqref{info-design-gmp-hetero:main} which is diagonal atomic. When specialized to $y^*$, this establishes the equivalence. 

It is sufficient to show that for every $\signal=\{\signal_{\state}: \, \state \in \allstates\}$ feasible for \eqref{info-design-gmp-hetero:main}, the  diagonal atomic $\signal^{\text{at}}:= \{\signal_{\state}^{\text{at}}(x)=\delta(x-\E_{\signal_{\state}}(x)): \state \in \allstates \}$ is also feasible and satisfies $J(\signal^{\text{at}}) \leq J(\signal)$. 

For every $y$,  
\ifthenelse{\equal{\numofcolumns}{1}}
{
\begin{equation*}
\begin{split}
J(\signal)-J(\signal^{\text{at}})& =\sum_{\state, i} \prior(\state) y_i \left(\int \congfunc_{\state,i}(x_i) \signal_{\state}(x) \, \de x - \int \congfunc_{\state,i}(x_i) \signal^{\text{at}}_{\state}(x) \, \de x\right) 
\\
& \hspace{0.5in}+ \sum_{\state, i} \prior(\state) \left(\int x_i \congfunc_{\state,i}(x_i) \signal_{\state}(x) \, \de x - \int x_i \congfunc_{\state,i}(x_i) \signal^{\text{at}}_{\state}(x) \, \de x\right)  \geq 0
\end{split}
\end{equation*}
}{
$
J(\signal)-J(\signal^{\text{at}})=\sum_{\state, i} \prior(\state) y_i \left(\int \congfunc_{\state,i}(x_i) \signal_{\state}(x) \, \de x - \int \congfunc_{\state,i}(x_i) \signal^{\text{at}}_{\state}(x) \, \de x\right) + \sum_{\state, i} \prior(\state) \left(\int x_i \congfunc_{\state,i}(x_i) \signal_{\state}(x) \, \de x - \int x_i \congfunc_{\state,i}(x_i) \signal^{\text{at}}_{\state}(x) \, \de x\right)\geq 0
$,  
}
where the inequality follows from Jensen's inequality due to convexity of $\congfunc_{\state,i}$ (since it is affine) and of $x_i \congfunc_{\state,i}$; it is easy to see that the latter follows from the convexity of $\congfunc_{\state,i}$.


\eqref{info-design-gmp-hetero:obedience} for $\signal$, $i=1$ and $j=2$ is:
\ifthenelse{\equal{\numofcolumns}{1}}
{
\begin{equation*}
\sum_{\state} \, \int \, \Bigg(\alpha_{1,\state,1} \, x_1^2  - \alpha_{1,\state,2} \,  x_{2} x_{1} + \alpha_{1,\state,1} \, x_1 y_1 - \alpha_{1,\state,2} \, x_1 y_2 + \alpha_{0,\state,1} \, x_{1} - \alpha_{0,\state,2} \, x_{1}
\Bigg) \, \signal_{\state}(x) \de x \,\prior(\state) \leq 0
\end{equation*}
}{
$
\sum_{\state}\int \, \Big(\alpha_{1,\state,1} \, x_1^2  - \alpha_{1,\state,2} \,  x_{2} x_{1} + \alpha_{1,\state,1} \, x_1 y_1 - \alpha_{1,\state,2} \, x_1 y_2+ \alpha_{0,\state,1} \, x_{1} - \alpha_{0,\state,2} \, x_{1}
\Big) \, \signal_{\state}(x) \de x \,\prior(\state) \leq 0
$
}
Plugging $x_2=\pfrac-x_1$, this is equivalent to:
\ifthenelse{\equal{\numofcolumns}{1}}
{
\begin{equation*}
\sum_{\state} \left(  (\alpha_{1,\state,1} + \alpha_{1,\state,2} ) \, \int  x_1^2 \, \signal_{\state}(x) \de x + (\alpha_{1,\state,1} y_1 + \alpha_{0,\state,1} - \pfrac \alpha_{1,\state,2} - y_2 \alpha_{1,\state,2} - \alpha_{0,\state,2})\int  x_1 \,  \signal_{\state}(x) \de x \right) \prior(\state) \leq 0
\end{equation*}
}{
$
\sum_{\state}\Big(  (\alpha_{1,\state,1} + \alpha_{1,\state,2} ) \, \int  x_1^2 \, \signal_{\state}(x) \de x + (\alpha_{1,\state,1} y_1 + \alpha_{0,\state,1}- \pfrac \alpha_{1,\state,2} - y_2 \alpha_{1,\state,2} - \alpha_{0,\state,2})\int  x_1 \,  \signal_{\state}(x) \de x \Big) \prior(\state) \leq 0
$
}
$\int x_1 \, \signal_{\state}(x) \de x = \int x_1 \, \signal^{\text{at}}_{\state}(x) \de x$ by definition, and $\int  x_1^2 \, \signal_{\state}(x) \de x \geq (\int x_1 \, \signal_{\state}(x) \de x)^2=\int  x_1^2 \, \signal^{\text{at}}_{\state}(x) \de x$ by Jensen's inequality. Therefore, 
\ifthenelse{\equal{\numofcolumns}{1}}
{
\begin{equation*}
\sum_{\state} \left(  (\alpha_{1,\state,1} + \alpha_{1,\state,2} ) \, \int  x_1^2 \, \signal^{\text{at}}_{\state}(x) \de x + (\alpha_{1,\state,1} y_1 + \alpha_{0,\state,1} - \pfrac \alpha_{1,\state,2} - y_2 \alpha_{1,\state,2} - \alpha_{0,\state,2})\int  x_1 \, \signal^{\text{at}}_{\state}(x) \de x \right) \prior(\state) \leq 0
\end{equation*}
}{
$
\sum_{\state}\Big(  (\alpha_{1,\state,1} + \alpha_{1,\state,2} ) \, \int  x_1^2 \, \signal^{\text{at}}_{\state}(x) \de x + (\alpha_{1,\state,1} y_1 + \alpha_{0,\state,1}- \pfrac \alpha_{1,\state,2} - y_2 \alpha_{1,\state,2} - \alpha_{0,\state,2})\int  x_1 \, \signal^{\text{at}}_{\state}(x) \de x \Big) \prior(\state) \leq 0
$
}
which is equivalent to \eqref{info-design-gmp-hetero:obedience} for $\signal^{\text{at}}$, $i=1$ and $j=2$. The proof for $i=2$ and $j=1$ is identical. 

The coefficients in $B_{\state}^{(i,j)}$ corresponding to the quadratic terms are zero and therefore $\int B_{\state}^{(i,j)} zz^T \signal_{\state}(x) \, \de x = \int B_{\state}^{(i,j)} zz^T \signal_{\state}^{\text{at}}(x) \, \de x$. Hence, $\signal^{\text{at}}$ satisfies \eqref{info-design-gmp-hetero:nash} trivially.

\subsubsection*{\underline{Equivalence between \eqref{eq:info-design-diagonal} and \eqref{gmp-to-sdp-hetero:main}}}
\eqref{eq:info-design-diagonal} is equivalent to:
\begin{equation}
\label{info-design-gmp-hetero-one-atomic:main}
\begin{split}
\min_{\hat{\signal} \in \hat{\allsignals}} & \, \, \int \, C \cdot \hat{Z} \, \de \hat{\signal} \\ 
\text{s.t. } & \quad \int \, A^{(i,j)} \cdot \hat{Z} \, \de \hat{\signal} \geq 0, \quad i, j \in [\nlinks]\\
& \quad \int \, B^{(i,j)} \cdot \hat{Z} \, \de \hat{\signal} \geq 0, \quad i, j \in [\nlinks]\\
& \quad \hat{\signal} \text{ is 1-atomic}
\end{split}
\end{equation}
where the expressions for the symmetric matrices $C$, $A^{(i,j)}$ and $B^{(i,j)}$ are in \ifthenelse{\equal{\isarxiv}{1}}
{Appendix~\ref{sec:matrix-expressions}}{\revisionchange{\cite[Appendix A]{Zhu.Savla:infodesign-arxiv20}}}, 
\ifthenelse{\equal{\numofcolumns}{1}}
{
\begin{equation*}
\hat{Z} = 
\begin{bmatrix}
1 & \hat{z}^T \\
 \hat{z} & \hat{z} \hat{z}^T
\end{bmatrix}, \,  \hat{z} = [x^{\state_1}_1, \ldots, x^{\state_1}_{\nlinks}, \ldots, x^{\state_{\nstates}}_1, \ldots, x^{\state_{\nstates}}_{\nlinks}, y_1, \ldots, y_{\nlinks}]^T
\end{equation*}
}{
$
\hat{Z} = 
\begin{bmatrix}
1 & \hat{z}^T \\
 \hat{z} & \hat{z} \hat{z}^T
\end{bmatrix}, \,  \hat{z} = [x^{\state_1}_1, \ldots, x^{\state_1}_{\nlinks}, \ldots, x^{\state_{\nstates}}_1, \ldots, x^{\state_{\nstates}}_{\nlinks}, y_1, \ldots, y_{\nlinks}]^T
$, 
}
and $\hat{\allsignals}$ is the set of probability distributions over $\hat{z}$ satisfying $x^{\state_k} \in \simplex(\pfrac)$ for all $k \in [\nstates]$ and $y \in \simplex(1-\pfrac)$.


It therefore suffices to establish the equivalence between \eqref{info-design-gmp-hetero-one-atomic:main} and \eqref{gmp-to-sdp-hetero:main}. We do this via a constrained version of \eqref{gmp-to-sdp-hetero:main}:
%
%
\begin{equation}
\label{gmp-to-sdp-hetero-rank-constraint:main}
\min_{M \succeq 0} \, \hat{J}(M) \quad \text{s.t. } \eqref{gmp-to-sdp-hetero:obedience}-\eqref{gmp-to-sdp-hetero:second-moment}, \, \,  \rank(M)=1
\end{equation}
Specifically, 
(a) for every $\hat{\signal}$ feasible for \eqref{info-design-gmp-hetero-one-atomic:main}, $M(\hat{\signal}) := \int \hat{Z} \, \de \hat{\signal}$ 
is feasible for \eqref{gmp-to-sdp-hetero-rank-constraint:main}, and hence also for \eqref{gmp-to-sdp-hetero:main}; 
(b) for every $M = \begin{bmatrix} 
1 & \hat{\mom}^T \\ \hat{\mom} & \hat{\mom} \hat{\mom}^T
\end{bmatrix}$ feasible for \eqref{gmp-to-sdp-hetero-rank-constraint:main}, $\hat{\signal}=\delta(\hat{z}-\hat{\mom})$ is feasible for \eqref{info-design-gmp-hetero-one-atomic:main}; and 
(c) there exists an optimal solution $M^*$ for \eqref{gmp-to-sdp-hetero:main} such that $\rank(M^*)=1$. (a) and (b) together imply the equivalence between \eqref{info-design-gmp-hetero-one-atomic:main} and \eqref{gmp-to-sdp-hetero-rank-constraint:main}, and (c) implies the equivalence between \eqref{gmp-to-sdp-hetero-rank-constraint:main} and \eqref{gmp-to-sdp-hetero:main}. The proofs are as follows.

\begin{enumerate}
\item[(a)] For a 1-atomic $\hat{\signal}$, $M(\hat{\signal})=\begin{bmatrix} 1 & \int \hat{z}^T \, \de \hat{\signal} \\ \int \hat{z} \, \de \hat{\signal} & \int \hat{z} \de \hat{\signal} \, \int \hat{z}^T \de \hat{\signal} \end{bmatrix}$ implying that $M(\hat{\signal})$ is rank one and positive semidefinite. $M(\hat{\signal})$ satisfying \eqref{gmp-to-sdp-hetero:obedience} and \eqref{gmp-to-sdp-hetero:nash} follow from the corresponding constraints in \eqref{info-design-gmp-hetero-one-atomic:main}. 
\eqref{gmp-to-sdp-hetero:unit-mass} follows from the definition of $M(\hat{\signal})$, and the rest of the constraints in \eqref{gmp-to-sdp-hetero:main} follow from constraints on the support of $\hat{\signal}$. 
\item[(b)] Proposition~\ref{prop:tms-realization} implies that the 1-atomic $\hat{\signal}=\delta(\hat{z}-\hat{\mom})$ belongs to $\hat{\allsignals}$. Simple algebra shows the equivalence between the other constraints in \eqref{info-design-gmp-hetero-one-atomic:main} and the corresponding constraints in \eqref{gmp-to-sdp-hetero:main}. 
\item[(c)] 
%
It is sufficient to show that, for every $M=\begin{bmatrix} 1 & \hat{\mom}^T \\ \hat{\mom} & M^0 \end{bmatrix}$ feasible for \eqref{gmp-to-sdp-hetero:main}, the rank one $\hat{M}=\begin{bmatrix} 1 & \hat{\mom}^T \\ \hat{\mom} & 
\hat{\mom} \hat{\mom}^T
\end{bmatrix}$
 is also feasible and satisfies $\hat{J}(M) \geq \hat{J}(\hat{M})$. 

$\hat{J}(M)-\hat{J}(\hat{M})=C^0 \cdot (M^0 - \hat{\mom} \hat{\mom}^T)$, where $C^0$ is the principal submatrix of $C$ obtained by removing the first row and the first column. $M \succeq 0$ implies $M^0 - \hat{\mom} \hat{\mom}^T \succeq 0$. It is easy to see that $C^0$ is positive semidefinite. Since the inner product of positive semidefinite matrices is non-negative, this implies that $\hat{J}(M)-\hat{J}(\hat{M}) \geq 0$. 

Feasibility of \eqref{gmp-to-sdp-hetero:unit-mass}-\eqref{gmp-to-sdp-hetero:non-negative} follows from the definition of $\hat{M}$. It is easy to see that $S_x^{(k)} \cdot \hat{M}= S_x^{(k)} \cdot M$ and $S_y \cdot \hat{M} = S \cdot M$, and therefore \eqref{gmp-to-sdp-hetero:simplex} is also satisfied. Also, for all $i \in [\nlinks]$ and $k \in [\nstates]$,
\ifthenelse{\equal{\numofcolumns}{1}}
{
\begin{equation*}
T_x^{(i,k)} \cdot \hat{M} = -\pfrac \hat{\mom}_{i} + (\onebf  e_i^T) \cdot (\hat{\mom} {\hat{\mom}}^T) = -\pfrac \hat{\mom}_{i} + \sum_{j} \hat{\mom}_{i} \hat{\mom}_{j} = -\pfrac \hat{\mom}_{i} + \pfrac \hat{\mom}_{i} = 0
\end{equation*}
}{
$
T_x^{(i,k)} \cdot \hat{M} = -\pfrac \hat{\mom}_{i} + (\onebf  e_i^T) \cdot (\hat{\mom} {\hat{\mom}}^T) = -\pfrac \hat{\mom}_{i} + \sum_{j} \hat{\mom}_{i} \hat{\mom}_{j}  = -\pfrac \hat{\mom}_{i} + \pfrac \hat{\mom}_{i} = 0
$. 
}
Similarly, $T_y^{(i)} \cdot \hat{M} = 0$ for all $i \in [\nlinks]$, implying \eqref{gmp-to-sdp-hetero:second-moment} is satisfied by $\hat{M}$. 
\eqref{gmp-to-sdp-hetero:obedience} for $M$ for $i=1$ and $j=2$ is:
\ifthenelse{\equal{\numofcolumns}{1}}
{
\begin{multline*}
\sum_{k} \left(\alpha_{1,\state_k,1} M^0_{2(k-1)+1,2(k-1)+1} - \alpha_{1,\state_k,2} M^0_{2(k-1)+1,2k} + \alpha_{1,\state_k,1} M^0_{2(k-1)+1,2 \nstates+1} - \alpha_{1,\state_k,2} M^0_{2(k-1)+1,2 \nstates+2} \right. \\ \left. + (\alpha_{0,\state_k,1} - \alpha_{0,\state_k,2}) \hat{\mom}_{2(k-1)+1} \right) \prior(\state_k) \leq 0
\end{multline*}
}{
$
\sum_{k}\Big(\alpha_{1,\state_k,1} M^0_{2(k-1)+1,2(k-1)+1} - \alpha_{1,\state_k,2} M^0_{2(k-1)+1,2k} + \alpha_{1,\state_k,1} M^0_{2(k-1)+1,2 \nstates+1} - \alpha_{1,\state_k,2} M^0_{2(k-1)+1,2 \nstates+2} + (\alpha_{0,\state_k,1} - \alpha_{0,\state_k,2}) \hat{\mom}_{2(k-1)+1} \Big) \prior(\state_k) \leq 0
$.
}
Plugging $M^0_{2(k-1)+1,2k} = \pfrac \hat{\mom}_{2(k-1)+1} - M^0_{2(k-1)+1,2(k-1)+1}$ and $M^0_{2(k-1)+1,2 \nstates+2} = (1-\pfrac) \hat{\mom}_{2(k-1)+1} - M^0_{2(k-1)+1,2 \nstates+1}$, this is equivalent to
\ifthenelse{\equal{\numofcolumns}{1}}
{
\begin{equation}
\label{eq:obed-crude}
\begin{split}
\sum_k \Big( (\alpha_{1,\state_k,1} + \alpha_{1,\state_k,2}) (M^0_{2(k-1)+1,2(k-1)+1} & +M^0_{2(k-1)+1,2 \nstates+1}) \\ & + (\alpha_{0,\state_k,1} - \alpha_{0,\state_k,2} -  \alpha_{1,\state_k,2}) \hat{\mom}_{2(k-1)+1} \Big) \prior(\state_k) \leq 0
\end{split}
\end{equation}
}{
\begin{equation}
\label{eq:obed-crude}
\begin{split}
\sum_k & \Bigg( (\alpha_{1,\state_k,1} + \alpha_{1,\state_k,2}) (M^0_{2(k-1)+1,2(k-1)+1} \\ & \quad +M^0_{2(k-1)+1,2 \nstates+1}) + (\alpha_{0,\state_k,1} - \alpha_{0,\state_k,2} \\ & \quad -  \alpha_{1,\state_k,2}) \hat{\mom}_{2(k-1)+1} \Bigg) \prior(\state_k) \leq 0
\end{split}
\end{equation}
}
\ifthenelse{\equal{\numofcolumns}{1}}
{
\begin{align*}
M \succeq 0 \, & \implies \, 
\begin{bmatrix}
1 & \hat{\mom}_{2(k-1)+1} & \hat{\mom}_{2\nstates+1} \\
* & M^0_{2(k-1)+1,2(k-1)+1} & M^0_{2(k-1)+1,2 \nstates+1} \\
* & * & M^0_{2 \nstates+1,2 \nstates+1}
\end{bmatrix} \succeq 0 \\ &  \implies \, 
\begin{bmatrix}
M^0_{2(k-1)+1,2(k-1)+1} & M^0_{2(k-1)+1,2 \nstates+1} \\
* & M^0_{2 \nstates+1,2 \nstates+1}
\end{bmatrix}  - 
\begin{bmatrix}
\hat{\mom}_{2(k-1)+1} \\
\hat{\mom}_{2\nstates+1}
\end{bmatrix} \, 
\begin{bmatrix}
\hat{\mom}_{2(k-1)+1} & \hat{\mom}_{2\nstates+1}
\end{bmatrix} \succeq 0
\end{align*}
}{
\begin{align*}
& M \succeq 0 \, \implies  
\\
&
\begin{bmatrix}
1 & \hat{\mom}_{2(k-1)+1} & \hat{\mom}_{2\nstates+1} \\
* & M^0_{2(k-1)+1,2(k-1)+1} & M^0_{2(k-1)+1,2 \nstates+1} \\
* & * & M^0_{2 \nstates+1,2 \nstates+1}
\end{bmatrix} \succeq 0 \\ &  \implies \, 
\begin{bmatrix}
M^0_{2(k-1)+1,2(k-1)+1} & M^0_{2(k-1)+1,2 \nstates+1} \\
* & M^0_{2 \nstates+1,2 \nstates+1}
\end{bmatrix}  \\
& \qquad \qquad - 
\begin{bmatrix}
\hat{\mom}_{2(k-1)+1} \\
\hat{\mom}_{2\nstates+1}
\end{bmatrix} \, 
\begin{bmatrix}
\hat{\mom}_{2(k-1)+1} & \hat{\mom}_{2\nstates+1}
\end{bmatrix} \succeq 0
\end{align*}
}
Inner product with \\ $\begin{bmatrix} \alpha_{1,\state_k,1} + \alpha_{1,\state_k,2} & \alpha_{1,\state_k,1} + \alpha_{1,\state_k,2} \\ 0 & 0 \end{bmatrix} \succeq 0$ gives
\ifthenelse{\equal{\numofcolumns}{1}}
{
\begin{equation*}
(\alpha_{1,\state_k,1} + \alpha_{1,\state_k,2}) (M^0_{2(k-1)+1,2(k-1)+1}+M^0_{2(k-1)+1,2 \nstates+1}) \geq (\alpha_{1,\state_k,1} + \alpha_{1,\state_k,2}) (\hat{\mom}^2_{2(k-1)+1} + \hat{\mom}_{2(k-1)+1} \hat{\mom}_{2 \nstates +1})
\end{equation*}
}{
$M^0_{2(k-1)+1,2(k-1)+1}+M^0_{2(k-1)+1,2 \nstates+1} \geq \hat{\mom}^2_{2(k-1)+1} + \hat{\mom}_{2(k-1)+1} \hat{\mom}_{2 \nstates +1}$
}
Plugging into \eqref{eq:obed-crude} implies that \eqref{gmp-to-sdp-hetero:obedience} is satisfied by $\hat{M}$ for $i=1$ and $j=2$. The proof for $i=2$ and $j=1$, as well as for \eqref{gmp-to-sdp-hetero:nash}, follows similarly. 
\end{enumerate}

\subsection{Proof of Theorem~\ref{thm:cost-monotonicity}}
\label{sec:proof-non-increasing}
For every $\pfrac \in [0,1]$, the feasible set for \eqref{eq:info-design-diagonal} is non-empty. Among the constraints that characterize the feasible set, the only ones which depend on $\pfrac$ are the linear equalities and inequalities associated with the characterization of $\simplex(\pfrac)$ and $\simplex(1-\pfrac)$, and are therefore continuous in $\pfrac$. Therefore, the feasible set is continuous in $\pfrac \in [0,1]$.\footnote{We forego excessive formalism in arguing about continuity of the feasible set and the optimal solution set with respect to $\pfrac \in [0,1]$. A formal argument would require to define these sets as point to set mappings and study the continuity of such mappings, e.g., see \cite[Definition 2.2.1]{Fiacco:83}, but would not add further insight.}  Furthermore, continuity of link latency functions implies that $J^{\diag}(x,y)$ is continuous. Therefore, \cite[Theorem 2.2.2]{Fiacco:83} implies that $J^{\diag,*}(\pfrac)$ is continuous and the set of optimal solutions to \eqref{eq:info-design-diagonal} is upper semi-continuous in $\pfrac \in [0,1]$, which in turn implies that there exists optimal solution $(x^*(\pfrac),y^*(\pfrac))$ which is continuous in $\pfrac \in [0,1]$. For such a solution, $C_{ij}^*(\pfrac) := \sum_{\state} \, \prior(\state) \, \left(\congfunc_{\state,_i}(x^{*,\state}_i(\pfrac) + y^*_i(\pfrac)) - \congfunc_{\state,j}(x^{*,\state}_j(\pfrac) + y^*_j(\pfrac))\right)$, $i, j \in [\nlinks]$, are also continuous. Consequently, almost every $\pfrac \in [0,1]$ belongs to a non-zero interval over which, for every $i, j \in [\nlinks]$, $C_{ij}^*(\pfrac)$ is either non-positive or positive. Since $J^{\diag,*}(\pfrac)$ is continuous, it suffices to show that  $J^{\diag,*}(\pfrac)$ is monotonically non-increasing over such intervals.

Consider one such interval $[\pfrac_1,\pfrac_2] \subseteq [0,1]$, and define the following over it, for $\epsilon \in \simplex_{\nlinks}$:
\ifthenelse{\equal{\numofcolumns}{1}}
{
\begin{equation}
\label{eq:traj-case1}
\begin{split}
x^{\state}_{i}(\epsilon,\pfrac) & = x^{*,\state}_{i}(\pfrac_1) + \epsilon_i (\pfrac - \pfrac_1), \quad 
y_i(\epsilon,\pfrac) = y^*_i(\pfrac_1) - \epsilon_i (\pfrac - \pfrac_1), \quad i \in [\nlinks], \, \state \in \Omega
\end{split}
\end{equation}
}{
\begin{equation}
\label{eq:traj-case1}
\begin{split}
x^{\state}_{i}(\epsilon,\pfrac) & = x^{*,\state}_{i}(\pfrac_1) + \epsilon_i (\pfrac - \pfrac_1), \\
y_i(\epsilon,\pfrac) & = y^*_i(\pfrac_1) - \epsilon_i (\pfrac - \pfrac_1), \quad i \in [\nlinks], \, \state \in \Omega
\end{split}
\end{equation}
}
\eqref{eq:traj-case1} implies that, for all $\epsilon \in \simplex_{\nlinks}$ and $\pfrac \in [\pfrac_1,\pfrac_2]$,
\ifthenelse{\equal{\numofcolumns}{1}}
{
\begin{equation}
\label{eq:C-constant}
\begin{split}
x^{\state}_{i}(\epsilon,\pfrac)+y_i(\epsilon,\pfrac) & =x^{*,\state}_{i}(\pfrac_1)+y^*_i(\pfrac_1), \qquad i \in [\nlinks], \, \state \in \Omega
\\
C_{ij}(\epsilon,\pfrac) & :=\sum_{\state} \prior(\state) (\congfunc_{\state,i}(x^{\state}_{i}(\epsilon,\pfrac)+y_i(\epsilon,\pfrac)) -\congfunc_{\state,j}(x^{\state}_{j}(\epsilon,\pfrac)+y_j(\epsilon,\pfrac))) 
\\
& = \sum_{\state} \prior(\state) (\congfunc_{\state,i}(x^{*,\state}_{i}(\pfrac_1)+y^*_i(\pfrac_1))-\congfunc_{\state,j}(x^{*,\state}_{j}(\pfrac_1)+y^*_j(\pfrac_1))) = C_{ij}^*(\pfrac_1), \quad i, \, j \in [\nlinks]
\end{split}
\end{equation}
}{
\begin{equation}
\label{eq:C-constant}
\begin{split}
&x^{\state}_{i}(\epsilon,\pfrac)+y_i(\epsilon,\pfrac)  =x^{*,\state}_{i}(\pfrac_1)+y^*_i(\pfrac_1), \qquad i \in [\nlinks], \, \state \in \Omega
\\
&C_{ij}(\epsilon,\pfrac)  :=\sum_{\state} \prior(\state) (\congfunc_{\state,i}(x^{\state}_{i}(\epsilon,\pfrac)+y_i(\epsilon,\pfrac)) \\ & \hspace{0.85in}-\congfunc_{\state,j}(x^{\state}_{j}(\epsilon,\pfrac)+y_j(\epsilon,\pfrac))) 
\\
& \hspace{0.5in}= \sum_{\state} \prior(\state) (\congfunc_{\state,i}(x^{*,\state}_{i}(\pfrac_1)+y^*_i(\pfrac_1)) \\ & \hspace{0.85in}-\congfunc_{\state,j}(x^{*,\state}_{j}(\pfrac_1)+y^*_j(\pfrac_1))) \\ & \hspace{0.5in}= C_{ij}^*(\pfrac_1), \quad i, \, j \in [\nlinks]
\end{split}
\end{equation}
}
\eqref{eq:info-design-diagonal:cost} and \eqref{eq:C-constant} imply that for all $\epsilon \in \simplex_{\nlinks}$ and $\pfrac \in [\pfrac_1,\pfrac_2]$,
\ifthenelse{\equal{\numofcolumns}{1}}
{
\begin{equation*}
\begin{split}
J^{\diag}(\epsilon,\pfrac) :=J^{\diag}(x(\epsilon,\pfrac),y(\epsilon,\pfrac)) = \sum_{i,\state} \prior(\state) (x^{*,\state}_{i}(\pfrac_1)+y^*_i(\pfrac_1)) \congfunc_{\state,i}(x^{*,\state}_{i}(\pfrac_1)+y^*_i(\pfrac_1)) = J^{\diag,*}(\pfrac_1)
\end{split}
\end{equation*}
}{
$
J^{\diag}(\epsilon,\pfrac)  :=J^{\diag}(x(\epsilon,\pfrac),y(\epsilon,\pfrac))= \sum_{i,\state} \prior(\state) (x^{*,\state}_{i}(\pfrac_1)+y^*_i(\pfrac_1)) \congfunc_{\state,i}(x^{*,\state}_{i}(\pfrac_1)+y^*_i(\pfrac_1))= J^{\diag,*}(\pfrac_1)
$. 
}
If $(x(\epsilon,\pfrac),y(\epsilon,\pfrac))$ is feasible, then $J^{\diag,*}(\pfrac) \leq J^{\diag}((x(\epsilon,\pfrac),y(\epsilon,\pfrac))) = J^{\diag,*}(\pfrac_1)$, thereby establishing the theorem. We now establish feasibility of $(x(\epsilon,\pfrac),y(\epsilon,\pfrac))$.

It is straightforward to check that $\sum_i x_i^{\state}(\epsilon,\pfrac) =\pfrac$ and $\sum_i y_i(\epsilon,\pfrac) =1-\pfrac$ for all $\epsilon \in [0,1]$ and $\state$.
By construction in \eqref{eq:traj-case1}, $x^{\state}_{i}(\epsilon,\pfrac) \geq x^{*,\state}_{i}(\pfrac_1) \geq 0$ for all $i, \, \state$, where the second inequality follows from optimality, and hence feasibility, of $x^{*,\state}_{i}(\pfrac_1)$. Noting from \eqref{eq:traj-case1} that $y_i(\epsilon,\pfrac)$, $i \in [\nlinks]$, is non-increasing in $\pfrac$, its non-negativity is ensured for all $\pfrac$ by ensuring non-negativity for $\pfrac=\pfrac_2$. This corresponds to choosing:
\begin{equation}
\label{eq:eps-choice}
\epsilon \in \setdef{\tilde{\epsilon} \in \simplex_{\nlinks}}{\tilde{\epsilon}_i \leq y_i^*(\pfrac_1)/(\pfrac_2-\pfrac_1), \, \, i \in [\nlinks]}
\end{equation}
The set in \eqref{eq:eps-choice} is non-empty because it contains $\epsilon=y^*(\pfrac_1)/(1-\pfrac_1)$. The feasibility of the inequalities in \eqref{eq:info-design-diagonal:obedience}-\eqref{eq:info-design-diagonal:nash} is established for a given $(i,j)$ by considering the sign of $C_{ij}^*(\pfrac_1)$ as follows.
\begin{itemize}
\item \underline{$C_{ij}^*(\pfrac_1) \leq 0$}. \eqref{eq:C-constant} implies $C_{ij}(\epsilon,\pfrac) \leq 0$, which in turn implies that \eqref{eq:info-design-diagonal:nash} hold true for $(i,j)$. Feasibility of \eqref{eq:info-design-diagonal:obedience} for $(i,j)$ also follows from 
\eqref{eq:C-constant}:
\ifthenelse{\equal{\numofcolumns}{1}}
{
\begin{equation*}
\begin{split}
& \sum_{\state} \prior(\state) \, x_i^{\state}(\epsilon,\pfrac) \left(\congfunc_{\state,i}(x_i^{\state}(\epsilon,\pfrac)+y_i(\epsilon,\pfrac)) - \congfunc_{\state,j}(x_j^{\state}(\epsilon,\pfrac)+y_j(\epsilon,\pfrac))\right)
\\
= & \sum_{\state}\prior(\state) \, x^{*,\state}_{i} \left(\congfunc_{\state,i}(x_i^{*,\state}(\pfrac_1)+y^*_i(\pfrac)) - \congfunc_{\state,j}(x_j^{*,\state}(\pfrac_1)+y^*_j(\pfrac_1))\right) + \epsilon_i (\pfrac - \pfrac_1) \, C_{ij}^*(\pfrac_1)
\\
& \leq \sum_{\state}\prior(\state) \, x^{*,\state}_{i} \left(\congfunc_{\state,i}(x_i^{*,\state}(\pfrac_1)+y^*_i(\pfrac)) - \congfunc_{\state,j}(x_j^{*,\state}(\pfrac_1)+y^*_j(\pfrac_1))\right) \leq 0
\end{split}
\end{equation*}
}{
$
\sum_{\state} \prior(\state) \, x_i^{\state}(\epsilon,\pfrac) \Big(\congfunc_{\state,i}(x_i^{\state}(\epsilon,\pfrac)+y_i(\epsilon,\pfrac)) - \congfunc_{\state,j}(x_j^{\state}(\epsilon,\pfrac)+y_j(\epsilon,\pfrac))\Big) = \sum_{\state}\prior(\state) \, x^{*,\state}_{i} \Big(\congfunc_{\state,i}(x_i^{*,\state}(\pfrac_1)+y^*_i(\pfrac)) - \congfunc_{\state,j}(x_j^{*,\state}(\pfrac_1)+y^*_j(\pfrac_1))\Big) + \epsilon_i (\pfrac - \pfrac_1) \, C_{ij}^*(\pfrac_1)\leq \sum_{\state}\prior(\state) \, x^{*,\state}_{i} \Big(\congfunc_{\state,i}(x_i^{*,\state}(\pfrac_1)+y^*_i(\pfrac))
- \congfunc_{\state,j}(x_j^{*,\state}(\pfrac_1)+y^*_j(\pfrac_1))\Big) \leq 0
$
}
where the last inequality follows from the feasibility of $(x^*(\pfrac_1),y^*(\pfrac_1))$.

\item 
\underline{$C_{ij}^*(\pfrac_1)>0$ and hence $y^*_i(\pfrac_1)=0$}. The only feasible solution to \eqref{eq:eps-choice} in this case is $\epsilon_i=0$. Therefore, $y_i(0,\pfrac)=y^*_i(\pfrac_1)=0$. \eqref{eq:C-constant} implies $C_{ij}(0,\pfrac)>0$ and hence \eqref{eq:info-design-diagonal:nash} is satisfied with equality for $(i,j)$.
Furthermore, 
\ifthenelse{\equal{\numofcolumns}{1}}
{
\begin{equation*}
\begin{split}
& \sum_{\state} \prior(\state) \, x_i^{\state}(0,\pfrac) \left(\congfunc_{\state,i}(x_i^{\state}(0,\pfrac)+y_i(0,\pfrac)) - \congfunc_{\state,j}(x_j^{\state}(0,\pfrac)+y_j(0,\pfrac))\right)
\\
= & \sum_{\state}\prior(\state) \, x^{*,\state}_{i} \left(\congfunc_{\state,i}(x_i^{*,\state}(\pfrac_1)+y^*_i(\pfrac)) - \congfunc_{\state,j}(x_j^{*,\state}(\pfrac_1)+y^*_j(\pfrac_1))\right) \leq 0
\end{split}
\end{equation*}
}{
$
\sum_{\state}\prior(\state) \, x_i^{\state}(0,\pfrac) \Big(\congfunc_{\state,i}(x_i^{\state}(0,\pfrac)+y_i(0,\pfrac)) - \congfunc_{\state,j}(x_j^{\state}(0,\pfrac)+y_j(0,\pfrac))\Big)=  \sum_{\state}\prior(\state) \, x^{*,\state}_{i} \Big(\congfunc_{\state,i}(x_i^{*,\state}(\pfrac_1)+y^*_i(\pfrac)) \\  - \congfunc_{\state,j}(x_j^{*,\state}(\pfrac_1)+y^*_j(\pfrac_1))\Big) \leq 0
$, 
}
which establishes \eqref{eq:info-design-diagonal:obedience} for $(i,j)$
\end{itemize}

\ifthenelse{\equal{\numofcolumns}{1}}
{\subsection{Technical Results}}{\subsection{Technical Result}}
\label{sec:technical-results}
%

We need additional definitions for the next result. These are adapted from \cite{Helton.Nie:12}. 
A \emph{truncated moment sequence (tms)} in $\tilde{\nlinks}$ variables and of degree $d$ is a finite sequence $t=(t_{a})$ indexed by nonnegative integer vectors $a:=(a_1,\ldots,a_{\tilde{\nlinks}}) \in \NN^{\tilde{\nlinks}}$ with $|a|:=a_1+\ldots+a_{\tilde{\nlinks}} \leq d$. 
Given a set $K$, a tms $t$ is said to admit a $K$- probability measure $\zeta$, i.e., a nonnegative Borel measure supported in $K$ with $\int_K \de \zeta = 1$, if 
\begin{equation*}
t_{a} = \int_K \hat{z}^{a} \, \de \zeta, \qquad \forall \, a \in \NN^{\tilde{\nlinks}}: \, |a| \leq d
\end{equation*}
where $\hat{z}^{a}=\hat{z}_1^{a_1} \ldots \hat{z}_{\tilde{\nlinks}}^{a_{\tilde{\nlinks}}}$ for $\hat{z}=(\hat{z}_1,\ldots,\hat{z}_{\tilde{\nlinks}})$. 

We are interested in tms of degree $2$. Accordingly, for brevity in notation, for $i,\, j \in [\tilde{\nlinks}]$, let 
\begin{equation}
\label{eq:tms-compact-notation}
t_i:=t_{(0,\ldots,0,\underbrace{1}_{i},0,\ldots,0)}, \, t_{i,j}:=t_{(0,\ldots,0,\underbrace{1}_{i},0,\ldots,0,\underbrace{1}_{j},0,\ldots,0)}
\end{equation}
We are also specifically interested in probability measures over the set of all $\hat{z}$ in $\RR^{\tilde{\nlinks}}$, $\tilde{\nlinks}=\nlinks (\nstates+1)$, satisfying $\sum_{i \in [(k-1)\nlinks+1:k\nlinks]} \hat{z}_i = \pfrac$ for all $k \in [\nstates]$ and 
$\sum_{i \in [\nstates \nlinks+1:(\nstates+1)\nlinks]} \hat{z}_i = 1- \pfrac$. Let the set of such probability measures be denoted as $\PP(\pfrac)$.

\ifthenelse{\equal{\isarxiv}{1}}
{
\begin{figure*}[!htb]
 \begin{center}
 \scalebox{.9}{
 \begin{minipage}[c]{.99\textwidth}
\begin{equation}
\label{eq:flat-ext-a}
M(w) :=\begin{bmatrix}
1 & w_1 & \ldots & w_{\tilde{\nlinks}} & w_{1,1} & \ldots & w_{1,\tilde{\nlinks}} & \ldots & w_{\tilde{\nlinks},1} & \ldots & w_{\tilde{\nlinks},\tilde{\nlinks}}\\
w_1 & w_{1,1} & \ldots & w_{1,\tilde{\nlinks}} & w_{1,1,1} & \ldots & w_{1,1,\tilde{\nlinks}} & \ldots & w_{1,\tilde{\nlinks},1} & \ldots & w_{1,\tilde{\nlinks},\tilde{\nlinks}}\\
\vdots & \vdots & \ldots & \vdots & \vdots & \ldots & \vdots & \ldots & \vdots & \ldots & \vdots\\
w_{\tilde{\nlinks}} & w_{\tilde{\nlinks},1} & \ldots & w_{\tilde{\nlinks},\tilde{\nlinks}} & w_{\tilde{\nlinks},1,1} & \ldots & w_{\tilde{\nlinks},1,\tilde{\nlinks}} & \ldots & w_{\tilde{\nlinks},\tilde{\nlinks},1} & \ldots & w_{\tilde{\nlinks},\tilde{\nlinks},\tilde{\nlinks}}\\
w_{1,1} & w_{1,1,1} & \ldots & w_{1,1,\tilde{\nlinks}} & w_{1,1,1,1} & \ldots & w_{1,1,1,\tilde{\nlinks}} & \ldots & w_{1,1,\tilde{\nlinks},1} & \ldots & w_{1,1,\tilde{\nlinks},\tilde{\nlinks}}\\
\vdots & \vdots & \ldots & \vdots & \vdots & \ldots & \vdots & \ldots & \vdots & \ldots & \vdots\\
w_{1,\tilde{\nlinks}} & w_{1,\tilde{\nlinks},1} & \ldots & w_{1,\tilde{\nlinks},\tilde{\nlinks}} & w_{1,\tilde{\nlinks},1,1} & \ldots & w_{1,\tilde{\nlinks},1,\tilde{\nlinks}} & \ldots & w_{1,\tilde{\nlinks},\tilde{\nlinks},1} & \ldots & w_{1,\tilde{\nlinks},\tilde{\nlinks},\tilde{\nlinks}}\\
\vdots & \vdots & \ldots & \vdots & \vdots & \ldots & \vdots & \ldots & \vdots & \ldots & \vdots\\
w_{\tilde{\nlinks},1} & w_{\tilde{\nlinks},1,1} & \ldots & w_{\tilde{\nlinks},1,\tilde{\nlinks}} & w_{\tilde{\nlinks},1,1,1} & \ldots & w_{\tilde{\nlinks},1,1,\tilde{\nlinks}} & \ldots & w_{\tilde{\nlinks},1,\tilde{\nlinks},1} & \ldots & w_{\tilde{\nlinks},1,\tilde{\nlinks},\tilde{\nlinks}}\\
\vdots & \vdots & \ldots & \vdots & \vdots & \ldots & \vdots & \ldots & \vdots & \ldots & \vdots \\
w_{\tilde{\nlinks},\tilde{\nlinks}} & w_{\tilde{\nlinks},\tilde{\nlinks},1} & \ldots & w_{\tilde{\nlinks},\tilde{\nlinks},\tilde{\nlinks}} & w_{\tilde{\nlinks},\tilde{\nlinks},1,1} & \ldots & w_{\tilde{\nlinks},\tilde{\nlinks},1,\tilde{\nlinks}} & \ldots & w_{\tilde{\nlinks},\tilde{\nlinks},\tilde{\nlinks},1} & \ldots & w_{\tilde{\nlinks},\tilde{\nlinks},\tilde{\nlinks},\tilde{\nlinks}}
\end{bmatrix} \succeq 0
\end{equation}
\end{minipage}
}
\end{center}
\end{figure*}
}

\begin{proposition}
\label{prop:tms-realization}
If a tms $t$ in $\tilde{\nlinks}=\nlinks(\nstates+1)$ variables and of degree $2$ satisfies:
\ifthenelse{\equal{\isarxiv}{1}}
{
\begin{subequations}
\label{eq:measure-admissible-sufficient}
\begin{equation}
\label{eq:measure-admissible-sufficient-moments}
}
{\begin{equation*}}
\begin{split}
& M(t)  :=
\begin{bmatrix}
1 & t_1 & \ldots & t_{\tilde{\nlinks}} \\
t_1 & t_{1,1} & \ldots & t_{1,\tilde{\nlinks}} \\
\vdots & \vdots & \ldots & \vdots \\
t_{\tilde{\nlinks}} & t_{\tilde{\nlinks},1} & \ldots & t_{\tilde{\nlinks},\tilde{\nlinks}}
\end{bmatrix} \succeq 0 
\\
& t_i \geq 0, \, \, i \in [\tilde{\nlinks}]; \qquad t_{i,j} \geq 0, \quad i, j \in [\tilde{\nlinks}] \\
& \sum_{i \in [(k-1)\nlinks+1:k \nlinks]} t_i = \pfrac, \, \, k \in [\nstates]; \qquad \sum_{i \in [\nstates \nlinks+1:(\nstates+1)\nlinks]} t_i = 1-\pfrac \\ 
& \sum_{j \in [(k-1)\nlinks+1:k\nlinks]} t_{i,j} = \pfrac t_i, \, \, i \in [(k-1)\nlinks+1:\nlinks], \, k \in [\nstates] \\
& \sum_{j \in [\nstates \nlinks+1:(\nstates+1)\nlinks]} t_{i,j} = (1-\pfrac) t_i, \, \, i \in [\nstates \nlinks+1:(\nstates+1)\nlinks]
\end{split}
\ifthenelse{\equal{\isarxiv}{1}}
{
\end{equation}
}
{\end{equation*}}

\ifthenelse{\equal{\isarxiv}{1}}
{
\begin{equation}
\label{eq:measure-admissible-sufficient-rank}}
{\begin{equation*}}
\rank(M(t)) = 1
\ifthenelse{\equal{\isarxiv}{1}}
{
\end{equation}
\end{subequations}}
{\end{equation*}}
then it admits a unique $\PP(\pfrac)$-probability measure, which is also 1-atomic and given by $\zeta(\hat{z})=\delta(\hat{z}-[t_1,\ldots,t_{\tilde{\nlinks}}]^T)$.
\end{proposition}
\ifthenelse{\equal{\isarxiv}{1}}
{
\begin{proof} 
\eqref{eq:measure-admissible-sufficient-rank} implies that 
\begin{equation}
\label{eq:moments-product-form}
t_{i,j} = t_i \, t_j, \qquad i, j \in [\tilde{\nlinks}]
\end{equation}
\cite[Theorem 1.1]{Helton.Nie:12}, which in turn is from \cite{Curto.Fialkow:05}, implies that a $t$ satisfying \eqref{eq:measure-admissible-sufficient-moments} admits a unique $\PP(\pfrac)$- probability measure if there exists a tms $w$ in $\tilde{\nlinks}$ variables and of degree $4$ such that it satisfies $w_a=t_a$ for all $|a| \leq 2$, and 
\eqref{eq:flat-ext-a}, \eqref{eq:flat-ext-b}:
\begin{equation}
\label{eq:flat-ext-b}
\begin{split}
& M_i(w) := \begin{bmatrix}
w_i & w_{i,1} & \ldots & w_{i,\tilde{\nlinks}} \\
w_{i,1} & w_{i,1,1} & \ldots & w_{i,1,\tilde{\nlinks}} \\
\vdots & \vdots & \ldots & \vdots \\
w_{i,\tilde{\nlinks}} & w_{i,\tilde{\nlinks},1} & \ldots & w_{i,\tilde{\nlinks},\tilde{\nlinks}}
\end{bmatrix} \succeq 0, \, \, i \in [\tilde{\nlinks}] \\
&\sum_{k \in [(l-1)\nlinks+1:l \nlinks]} w_{i,j,k} = \pfrac \, w_{i,j}, \quad i, j \in [\tilde{\nlinks}]
\\
& \sum_{k \in [\nstates \nlinks+1:(\nstates+1) \nlinks]} w_{i,j,k} = (1-\pfrac) \, w_{i,j}, \quad i, j \in [\tilde{\nlinks}], \, l \in [\nstates]
\end{split}
\end{equation}
where $w_i$, $w_{i,j}$, $w_{i,j,k}$, and $w_{i,j,k,l}$ are defined similar to \eqref{eq:tms-compact-notation}. 
Let
\begin{equation}
\label{eq:w-def}
w_{i,j,k} = t_i t_j t_k, \qquad w_{i,j,k,l} = t_i t_j t_k t_l, \qquad i, j, k, l \in [\tilde{\nlinks}]
\end{equation}
%
\eqref{eq:moments-product-form} and \eqref{eq:w-def} imply $w_{i,j,k}=t_{i,j} t_k=w_{i,j} t_k$, and therefore, $\sum_{k \in [(l-1)\nlinks+1:l \nlinks]} w_{i,j,k}=w_{i,j} \sum_{k \in [(l-1)\nlinks+1:l \nlinks]} t_k=\pfrac w_{i,j}$ for all $l \in [\nstates]$ and $\sum_{k \in [\nstates \nlinks+1:(\nstates+1) \nlinks]} w_{i,j,k}=w_{i,j} \sum_{k \in [\nstates \nlinks+1:(\nstates+1) \nlinks]} t_k=(1-\pfrac) w_{i,j}$. 
\eqref{eq:w-def} implies that every column of $M_i(w)$ is a multiple of the first column, and therefore $\rank(M_i(w)) = 1$. Since the leading entry $w_i$ is nonnegative, $M_i(w)$ is positive semidefinite. Along the same lines, $M(w)$ has rank one and is positive semidefinite.

Since $\rank(M(w))=1=\rank(M(t))$, \cite[Theorem 1.1]{Helton.Nie:12} implies that the unique probability measure $\zeta$ is 1-atomic. The expression for $\zeta$ is then trivial from the fact that $\E_{\zeta}[\hat{z}]=[t_1, \ldots,t_{\tilde{\nlinks}}]^T$.
%
\end{proof}}
{
\revisionchange{The proof, which can be found in the extended version~\cite{Zhu.Savla:infodesign-arxiv20}, is omitted here for lack of space.}
}

\end{document}